\documentclass[3p,authoryear]{elsarticle}

\usepackage{amsmath,amssymb,amsthm}

\newtheorem{proposition}{Proposition}[section]

\newtheorem{corollary}[proposition]{Corollary}
\newtheorem{definition}[proposition]{Definition}
\newtheorem{example}[proposition]{Example}
\newtheorem{hypothesis}[proposition]{Hypothesis}
\newtheorem{lemma}[proposition]{Lemma}
\newtheorem{question}[proposition]{Question}
\newtheorem{remark}[proposition]{Remark}
\newtheorem{simulation}[proposition]{Simulation}
\newtheorem{theorem}[proposition]{Theorem}

\usepackage{subcaption}

\usepackage[draft=false, hyperindex=false, pagebackref=true, linkcolor=blue, citecolor=blue, colorlinks=true]{hyperref}

\renewcommand*{\backref}[1]{}
\renewcommand*{\backrefalt}[4]{%
  \ifcase #1 %
    No citations.% use \relax if you do not want the "No citations" message
  \or
    (Cited on page~#2).%
  \else
    (Cited on pages~#2).%
  \fi%
}

\hypersetup{
 pdfauthor = {Isaac Vikram Chenchiah, Patrick D. Shipman},
 pdftitle = {Discrete Growing Systems and their Continuum Limits},
}

\biboptions{comma}

\newcommand{\C}{ \ensuremath{ \mathcal{C} } }
\renewcommand{\d}{ \ensuremath{ \mathrm{d} } }
\newcommand{\Def}{ \ensuremath{\R^{D \times D}} }
\newcommand{\e}{ \ensuremath{ \epsilon } }

\newcommand{\Energy}{ \ensuremath{ \mathcal{W} } }
\newcommand{\F}{ \ensuremath{ \mathcal{F} } }
\newcommand{\g}{ \ensuremath{ \gamma } }
\renewcommand{\l}{ \ensuremath{ \lambda } }
\renewcommand{\L}{ \ensuremath{ \mathcal{L} } }
\newcommand{\N}{ \ensuremath{ \mathbb{N} } }

\newcommand{\R}{ \ensuremath{ \mathbb{R} } }
\newcommand{\W}{ \ensuremath{ \widetilde{W} } }
\newcommand{\Z}{ \ensuremath{ \mathbb{Z} } }

\begin{document}

\begin{frontmatter}

\title{An Energy-Deformation Decomposition for Morphoelasticity}

\author[Isaac's address]{Isaac Vikram Chenchiah}
\address[Isaac's address]{School of Mathematics, University of Bristol, University Walk, Bristol BS8~1TW, UK}
\ead{Isaac.Chenchiah@bristol.ac.uk}
\ead[url]{http://www.maths.bris.ac.uk/people/profile/maivc/} 
 
\author[Patrick's address]{Patrick D. Shipman}
\address[Patrick's address]{Department of Mathematics, Colorado State University, 1874 Campus Delivery, Fort Collins, CO 80523-1874, USA}
\ead{shipman@math.colostate.edu}
\ead[url]{http://www.math.colostate.edu/~shipman}

\begin{abstract}
Mathematical models of biological growth commonly attempt to distinguish deformation due to growth from that due to mechanical stresses through a hypothesised multiplicative decomposition of the deformation gradient. Here we demonstrate that this hypothesis is fundamentally incompatible with shear-resistance and thus cannot accurately describe growing solids. Shifting the focus away from the kinematics of growth to the mechanical energy of the growing object enables us to propose an ``energy-deformation decomposition'' which accurately captures the influence of growth on mechanical energy. We provide a proof and computational verification of this for tissues with crystalline structure. Our arguments also apply to tissues with network structure. Due to the general nature of these results they apply to a wide range of models for growing systems.
\end{abstract}


\end{frontmatter}

%--------
\section{Introduction}

%--------
\subsection{Background}

Biological growth---of cells, tissues, organs and organisms---leads to morphological change as well as mechanical stresses such as tension in arteries (see, for example,~\cite{Holzapfel:2010}) and plant stems (for example,~\cite{Goriely:2006,Vandiver:2008}).  These play an important role in biological function.  Modeling biological growth and the accompanying mechanical stresses is of increasing interest in the biological, continuum mechanics and mathematical communities as it becomes abundantly clear that mechanical stresses are not only passive responses to growth, but also feed back to influence morphological development as well as biochemical pathways; see, for example,~\cite{Lintilhac:1984,Lynch:1997,Huang:2004}.  

Continuum models for growth typically draw from approaches first developed in the context of non-biological continua, for example, plasticity theory or mixture theory; see~\cite{Goriely:2008,Ambrosi:2011} for reviews. A very popular current approach is based on a multiplicative decomposition of the deformation gradient.  The idea, introduced by~\cite{Rodriguez:1994}, is as follows:  Suppose that a map $\varphi$ describes the deformation of a body from a reference configuration to a current configuration.  In the absence of growth, the elastic energy density of the current configuration is a function $W$ of the deformation gradient  $F := D \varphi$.  If, however, growth has also contributed to the deformation,  the mechanical energy is a function of only that part of the deformation resulting in elastic stresses.  Analogously to a standard approach in plasticity theory (\cite{Lee:1969}), the approach of multiplicative decomposition posits that the gradient $F$ is a product 
\begin{subequations} \label{eq:md}
\begin{equation} \label{eq:md1}
F = A \ G
\end{equation} 
of  tensors $G$, arising from growth, and $A$, arising from elastic deformation. (The decomposition~\eqref{eq:md1} is understood to hold point-wise in space-time.) There is not, in general, a deformation of the body with gradient $G$; that is, growth may not be compatible with any actual deformation.  The elastic deformation associated with $A$ restores compatibility (that is, the gradient nature) of $F$.    The elastic energy is a function only of $A$: 
\begin{equation} \label{eq:md2}
W = W(F G^{-1}).   
\end{equation}
\end{subequations}

%--------
\subsection{Main results}

From~\eqref{eq:md2}, the multiplicative decomposition can be understood as saying that energetically growth is a \emph{motion} (in its domain, by $G$) of the energy density $W$. Here we present a (in our view) compelling argument that this cannot be true in the presence of shear resistance (see Definition~\ref{def:shears} below) and non-dilational growth (defined in Section~\ref{sec:notation} below). Moreover we also propose a way to overcome this: that in addition to the deformation gradient being decomposed (multiplicatively), the energy density be decomposed (additively). 

This is done in Theorems~\ref{thm:Z^2_2} and~\ref{thm:Z^D}, and Corollary~\ref{cor:Z^D}. These focus on growing systems with a crystalline structure; we explain the reason for this below. However our insights and arguments are general and we discuss growing networks and growing continua in Section~\ref{sec:continua}. We highlight three features of our approach:

First, we study continua via discrete systems: We begin with discrete systems because they provide a context in which the concept of growth can be clearly and unambiguously formulated, and where the interaction of growth with deformation can be rigorously derived. To understand growth on a continuum level, we compare the continuum limits of the initial and the grown system. (However, in some instances, discrete systems themselves are appropriate models for growth with the discrete elements corresponding, for example, to biological cells, see~\cite{Odell:1981,Weliky:1990,Munoz:2010}.)

Second, we shift attention away from the kinematics of growth, cf.,~\eqref{eq:md1}, to the energetics of growth, cf.,~\eqref{eq:md2}. We submit that not only is this a better conceptual approach to growth but also facilitates the use of more sophisticated mathematical tools such as discrete-to-continuum limits and, more generally, the tools of the calculus of variations. In addition it enables a closer integration of the mechanics of growth with the biochemistry of growth, and thus permits the development of more holistic models of biological growth.

Third, we recognise that there are different kinds of growth processes, two of which we discuss in Section~\ref{sec:springs} below.

%--------
\subsection{Organisation of the paper}

We begin, in Section~\ref{sec:springs}, by considering \emph{growable springs}. These are one-dimensional elastic objects which are able to change their rest-length by a non-elastic process, namely \emph{growth}. There is more than one way in which this can occur; we present two ways, which we name \emph{replication} (Section~\ref{sec:replication}) and \emph{recombination} (Section~\ref{sec:recombination}). Of these replication, which involves mass transfer, more naturally corresponds to an intuitive understanding of growth. However, recombination, which is a constant-mass process, may be thought of as a one-dimensional conceptualisation of the biological process sometimes called ``remodelling'', see, for example,~\cite{Taber:1995,Taber:2001,Ambrosi:2011,Menzel:2012}. Recombination is also simpler to analyse mathematically so we focus on it in this paper. However, as a reader who follows our arguments can easily see, our claim that the multiplicative decomposition cannot adequately describe growing systems is true for generic growth processes including replication. Similarly, while we develop our alternate approach of ``energy-deformation decomposition'' in the context of recombination, the insight underlying it applies to generic growth processes (albeit the resulting formulation might be less elegant for replication).

In Section~\ref{sec:lattices} we introduce (node-spring) \emph{lattices}, which provide our model of a biological system. We do this, not because lattices are good models for biological systems but rather because they present the simplest context in which we can communicate our insights and develop our arguments. It would have been biologically more natural to use (node-spring) \emph{networks} but we judged that the resulting need for more detailed mathematical analysis would have obscured rather than clarified the essential features. After this, in Section~\ref{sec:1D}, we introduce a one-dimensional example which serves as a concrete context in which to introduce the questions that would occupy our attention.

Next we come to the heart of the paper: In Section~\ref{sec:homogeneous_lattices_recombining_homogeneously} we present our main results (in the context of homogeneous lattices growing homogeneously) and arguments (which are much broader in scope). Section~\ref{sec:homogenisation} extends these results to inhomogeneous situations by demonstrating that the energy-deformation decomposition is stable under homogenisation, both of growth and of elasticity. While Sections~\ref{sec:homogeneous_lattices_recombining_homogeneously} and~\ref{sec:homogenisation} considered continuum limits of lattices, in Section~\ref{sec:continua} we briefly touch upon the application of our insights to networks and continua.

%--------
\subsection{Notation and definitions}
\label{sec:notation}

$\Z$ is the set of integers, $\N$ the set of positive integers and $\R_+$ is the set of non-negative real numbers.
We set
\begin{align*}
\Z^D_N
 &:= \Z^D \cap [0,N]^D, \\
\partial \Z^D_N 
 &:= \Z^D \cap \partial [0,N]^D,
\end{align*} 
where $D$ is the dimension of the space and $N \in \N$. We denote the standard Euclidean basis in $\R^D$ by $\{e_1, e_2, \dots e_D\}$. For brevity we set 
\begin{align*}
e_{i \pm j}
 &:= e_i \pm e_j, \\
e_{i \pm j \pm k}
 &:= e_i \pm e_j \pm e_k
\end{align*} 
for $i,j,k = \{1,2,3\}$. $\langle \cdot, \cdot \rangle$ is the standard Euclidean inner product on $\R^D$.
For $\theta \in \R$,
\begin{equation*}
R_\theta := \begin{pmatrix} \cos \theta & -\sin \theta \\ \sin \theta & \cos \theta \end{pmatrix} \in SO(2).
\end{equation*}
We denote the identity matrix (in any dimension) by $I$. An \emph{isotropic} matrix is one that is a multiple of the identity. A \emph{dilation} is a deformation whose gradient is isotropic.

The cardinality of a set $S$ is $|S|$.

We use the term \emph{shear}, and the related phrases \emph{vanishes on a shear} and \emph{shear-resisting}, in a specific sense in this paper:

\begin{definition}[Shears] \label{def:shears}
(A non-trivial deformation gradient) $F \in \Def \setminus SO(D)$, is a shear if
\begin{equation*}
\| \cdot \| \circ F
 = \| \cdot \|
\end{equation*}
on a basis of $\R^D$. That is, there exists a basis $\{ v_1, \cdots, v_D \}$ of $\R^D$ such that
\begin{equation*}
\| F v_i \|
 = \| v_i \| \quad \forall i = 1, \dots, D.
\end{equation*}
Now let $W \colon \Def \to \R$ with $W(I) = \inf W$ .
\begin{enumerate}
\item $W$ vanishes on a shear if there exists a shear $F$ such that $W(F) = W(I)$.
\item $W$ is shear-resisting (or resists shears) if there is no shear $F$ for which $W(F) = W(I)$.
\end{enumerate}
\end{definition} 

Note that if $F \in \Def$ is a shear then so is $SO(2) \ F \ SO(2)$.

%--------
\section{Growable springs}
\label{sec:springs}

Our discrete systems are built of \emph{growable springs} which are one-dimensional objects which are characterised by a rest-length, $\ell$, and a current length, $e$. In addition to \emph{elastic deformation}, which changes $e$ but not $\ell$, the springs can also \emph{grow}, i.e., change $\ell$. 

The elastic energy $\W$ of such a spring depends on both $\ell$ and $e$. The nature of this dependence itself depends on the nature of the growth process being considered. In a moment we consider two kinds of growth processes, replication and recombination. However elasticity motivates the following constraints on $\W$:

We require that when $e=\ell$, $\W$ attains its minimum, which for convenience we set to be $0$. Further we ask that this state be a stable elastic equilibrium:
\begin{subequations} \label{eq:assumptions}
\begin{align}
\W(\ell,\ell) & = 0, \\
\frac{\partial}{\partial e} \W(\ell,\ell) &= 0, \\
\frac{\partial^2}{\partial e^2} \W(\ell,\ell) &\geqslant 0.
\end{align}
\end{subequations}
More generally we require that $\W(\ell,\cdot)$ be convex.

Note that $\W$ is the elastic energy stored in the spring. Any energy expended in growth can be accounted for by the inclusion of a \emph{growth energy}
which can be added to $\W$ to yield the total energy.

Next we consider two kinds of growth processes which we name replication and recombination. Although an analysis of mass-transfer is outside the scope of this work it is natural to think of recombination as a mass-preserving process and replication as changing the mass of the system. However here we are concerned, not with the underlying mechanisms which accomplish growth, but only with the interaction of growth (of either kind) with elasticity.

%--------
\subsection{Replication}
\label{sec:replication}

We introduce replication through an example; see Figure~\ref{fig:replication}.
\begin{figure}
\begin{center}
\includegraphics[scale=0.3]{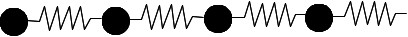}\\
$\downarrow$\\
\includegraphics[scale=0.3]{1D-Spring-System-1.jpg}\includegraphics[scale=0.3]{1D-Spring-System-1.jpg}
\caption{Replication}
\label{fig:replication}
\end{center}
\end{figure}

Consider two systems. The first, the \emph{initial system} consists of a single spring of rest length $\ell$. Its elastic energy is $W_i(\cdot) = \W(\ell,\cdot)$. The other, the \emph{grown system}, consists of two such springs attached in series. Its (total) rest-length is $2\ell$. Let its elastic energy be $W_g = \W(2\ell,\cdot)$.  

Let the grown system be elastically deformed to a (total) length of $2e$. Since $\W(\ell,\cdot)$ is convex, an energetically optimal state of the grown system is when each spring deforms to length $e$. That is,
\begin{equation*}
W_g(2e) = \W(2\ell,2e) = 2 \W(\ell,e) = 2W_i(e).
\end{equation*}
In other words, the grown system is energetically equivalent to two copies of the initial system. We call the corresponding growth process (i.e., the process that transforms the initial system into the grown system) a \emph{replication}. More generally:

\begin{definition}[Replication]
A replication is a growth process for which the elastic energy is a (positively) $1$-homogeneous function: For all $\eta, e, \ell >0$,
\begin{subequations}
\begin{equation} \label{eq:replication}
\W(\eta \ell,\eta e) = \eta \W(\ell,e).
\end{equation}
Let $W(\cdot) := \W(1,\cdot)$. From the $1$-homogenity of $\W$,
\begin{equation} \label{eq:replication-energy}
\W(\ell,e) = \ell \ W(\frac{e}{\ell}).
\end{equation}
\end{subequations}
\end{definition}

Note that the mass of the grown system is $\ell$ times the mass of the initial system.

\begin{remark}
Note that $1$-homogenity requires that the minimum elastic energy be $0$:
\begin{equation*}
\min_{\ell,e} \W(\ell,e) = \W(\ell,\ell) = \ell W(1).
\end{equation*}
Since this is true for all $\ell>0$ we conclude that $W(1) = 0$.

Moreover convexity of $\W(\ell,\cdot)$ implies convexity of $W$. From~\eqref{eq:assumptions},
\begin{alignat*}{2}
\frac{\partial}{\partial e} \W(\ell,\ell) &= 0
 &\implies W'(1) &= 0, \\
\frac{\partial^2}{\partial e^2} \W(\ell,\ell) &\geqslant 0
 &\implies W''(1) &\geqslant 0. 
\end{alignat*}
Note that
\begin{align*}
\frac{\partial}{\partial \ell} \W(\ell,\ell) &= W(1) - W'(1) = 0, \\
\frac{\partial^2}{\partial \ell^2} \W(\ell,\ell) &= \frac{1}{\ell} W''(1) \geqslant 0.
\end{align*}
Thus $e=\ell$ is a stable equilibrium for the growth process.
\end{remark}

\begin{example}[Hookean springs]
A Hookean spring has energy
\begin{equation*}
\frac{1}{2} k(\ell) (e - \ell)^2.
\end{equation*}
where $k(\ell)$ is the spring constant of a spring of rest-length $\ell$. Considering (as above) such springs joined in series leads to the well-known scaling $k(\ell) = \ell^{-1} k(1)$. Thus
\begin{equation*}
\W(\ell,e) = \ell \frac{1}{2} k(1) (\frac{e}{\ell} - 1)^2,
\end{equation*}
which corresponds to
\begin{equation*}
W(\cdot) = \frac{1}{2} k(1) (\cdot - 1)^2.
\end{equation*}
\end{example}

\begin{remark}[Linearisation]
Let $\ell_o$ be the initial length of a spring. For small growth $\ell = \ell_o + \g \ell_o$ (with $\g \ll 1$) and small deformation $e = \ell_o + \e \ell_o$ (with $\e \ll 1$) we have
\begin{equation*}
\W(\ell_o + \g \ell_o,\ell_o + \e \ell_o)
 = (\ell_o + \g \ell_o) W \left( \frac{\ell_o + \e \ell_o}{\ell_o + \g \ell_o} \right) \\
 \approx \ell_o W(1+\e-\g) \\
 \approx \frac{\ell_o}{2} W''(1) (\e - \g)^2.
\end{equation*}
\end{remark}

%--------
\subsection{Recombination}
\label{sec:recombination}

As with replication, we introduce recombination through an example; see Figure~\ref{fig:recombination}.
\begin{figure}
\begin{center}
\includegraphics[scale=0.25]{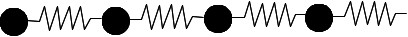}\\
\includegraphics[scale=0.25]{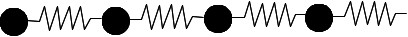}\\
$\downarrow$\\
\includegraphics[scale=0.25]{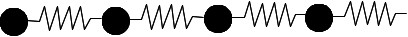}\includegraphics[scale=0.3]{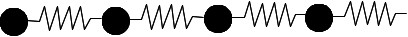} \vspace{1cm} \\
\includegraphics[scale=0.25]{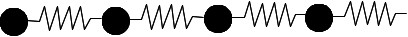}\\
\includegraphics[scale=0.25]{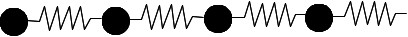}\\
\includegraphics[scale=0.25]{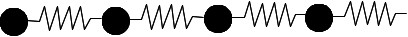}\\
$\downarrow$\\
\includegraphics[scale=0.25]{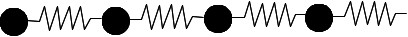}\includegraphics[scale=0.3]{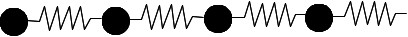}\includegraphics[scale=0.3]{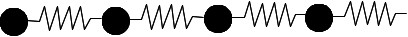} \vspace{1cm} \\
\includegraphics[scale=0.25]{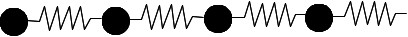}\\
\includegraphics[scale=0.25]{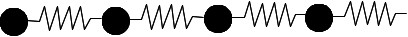}\\
\includegraphics[scale=0.25]{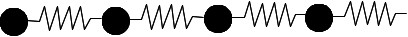}\\
\includegraphics[scale=0.25]{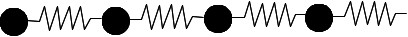}\\
$\downarrow$\\
\includegraphics[scale=0.25]{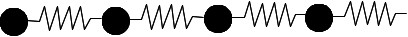}\includegraphics[scale=0.3]{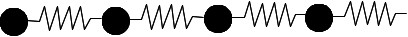}\\
\includegraphics[scale=0.25]{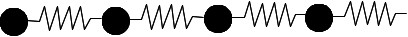}\includegraphics[scale=0.3]{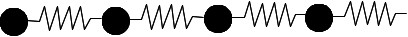}\\
$\downarrow$\\
\includegraphics[scale=0.25]{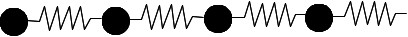}\includegraphics[scale=0.3]{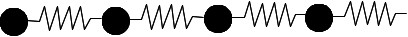}\includegraphics[scale=0.3]{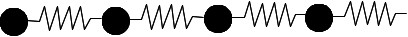}\includegraphics[scale=0.3]{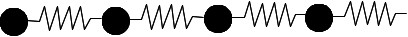}
\caption{Recombinations}
\label{fig:recombination}
\end{center}
\end{figure}

Consider two systems. The first, the \emph{initial} system consists of two springs in parallel, each of rest length $\ell$. Its elastic energy is $W_i(\cdot) = 2\W(\ell,\cdot)$.

The other, the \emph{grown} system also consists of two such springs, but attached in series. Its (total) rest-length is $2\ell$. Let its elastic energy be $W_g(\cdot) = \W(2\ell,\cdot)$.

Let the grown system be elastically deformed to a (total) length of $2e$. As before, the optimal configuration is for each spring to deform to length $e$. Thus,
\begin{equation*}
W_g(2e) = \W(2\ell,2e) = 2 \W(\ell,e) = W_i(e).
\end{equation*}
In other words, the grown system is energetically equivalent to the initial system. We call the corresponding growth process a \emph{recombination}. More generally:

\begin{definition}[Recombination]
A recombination is a growth process for which the elastic energy is a (positively) $0$-homogeneous function: For all $\eta, e, \ell >0$,
\begin{subequations}
\begin{equation} \label{eq:recombination}
\W(\eta \ell,\eta e) = \W(\ell,e).
\end{equation}
Let $W(\cdot) := \W(1,\cdot)$. From the $0$-homogenity of $\W$,
\begin{equation} \label{eq:recombination-energy}
\W(\ell,e) = W(\frac{e}{\ell}).
\end{equation}
\end{subequations}
\end{definition}

\begin{remark}
As before convexity of $\W(\ell,\cdot)$ implies convexity of $W$, and, from~\eqref{eq:assumptions},
\begin{alignat*}{2}
\W(\ell,\ell) & = 0 
 &\implies W(1) & = 0, \\
\frac{\partial}{\partial e} \W(\ell,\ell) &= 0
 &\implies W'(1) &= 0, \\
\frac{\partial^2}{\partial e^2} \W(\ell,\ell) &\geqslant 0
 &\implies W''(1) &\geqslant 0. 
\end{alignat*}
Note that
\begin{align*}
\frac{\partial}{\partial \ell} \W(\ell,\ell) &= - \frac{1}{\ell} W'(1) = 0, \\
\frac{\partial^2}{\partial \ell^2} \W(\ell,\ell) &= \frac{1}{\ell^2} W''(1) \geqslant 0.
\end{align*}
Thus $e=\ell$ is a stable equilibrium for the growth process.
\end{remark}

\begin{example}[Hookean springs]
Consider a rearrageable spring system which consists of $N$ Hookean springs, each with spring constant $k(\ell_o)$ and rest-length $\ell_o$. When the rest-length of the system is $\ell < N \ell_o$, the system consists of $\frac{N \ell_o}{\ell}$ parallel systems of $\frac{\ell}{\ell_o}$ springs in series. When the total length of the system is $e$ the energy of the system is 
\begin{align*}
\W(\ell,e) &= \frac{N}{2} k(\ell_o) \ell_o^2 (\frac{e}{\ell} - 1)^2
\end{align*}
which corresponds to
\begin{equation*}
W(\cdot) = \frac{N}{2} k(\ell_o) (\cdot - 1)^2.
\end{equation*}
\end{example}

\begin{remark}[Linearisation]
Let $\ell_o$ be the initial length of a spring. For small growth $\ell = \ell_o + \g \ell_o$ (with $\g \ll 1$) and small deformation $e = \ell_o + \e \ell_o$ (with $\e \ll 1$) we have
\begin{equation*}
\W(\ell_o + \g \ell_o,\ell_o + \e \ell_o)
 = W \left( \frac{\ell_o + \e \ell_o}{\ell_o + \g \ell_o} \right) \\
 \approx W(1+\e-\g) \\
 \approx \frac{1}{2} W''(1) (\e - \g)^2.
\end{equation*}
\end{remark}

%--------
\subsection{Generic growth}

Replication and recombination are special cases of growth processes that can be represented by a (positively) $p$-homogeneous energy: 
\begin{equation*}
\W(\ell,e) = \ell^p W(\frac{e}{\ell})
\end{equation*}
with recombination corresponding to $p=0$ and replication to $p=1$.

When the rest-length of such a spring changes from $\ell_i$ to $\ell_g$ the elastic energy density changes from
\begin{equation*}
W_i(\cdot)
 = \ell_i^p W(\frac{\cdot}{\ell_i})
\end{equation*} 
to
\begin{equation*}
W_g(\cdot)
 = \ell_g^p W(\frac{\cdot}{\ell_g}).
\end{equation*}
Thus
\begin{equation*}
W_g(\cdot)
 = \ell_g^p W(\frac{\cdot}{\ell_g})
 = \left( \frac{\ell_g}{\ell_i} \right)^p \ell_i^p W \left( \frac{\cdot}{\ell_i} \frac{\ell_i}{\ell_g} \right)
 = g^p W_i \left( \cdot g^{-1} \right).
\end{equation*}
where $g = \frac{\ell_g}{\ell_i} > 0$ is the \emph{growth} of the spring. Note that we allow $g \in (0,1)$.

%--------
\section{Preliminaries}

%--------
\subsection{Lattices}
\label{sec:lattices}

While it excludes, for example, triangular and hexagonal unit cells, the following restricted definition of lattices suffices for our purposes:

\begin{definition}[Lattices]
A lattice is a triple $(\Z^D,\C,\L)$ where
\begin{enumerate}
\item the connectivity $\C \subset \Z^D$ is finite and satisfies
\begin{equation} \label{eq:local10}
v \in \C \implies -v \notin \C,
\end{equation}
and
\item the rest-length $\L$ is a function 
\begin{equation*}
\left\{ (x,y) \in \Z^D \times \Z^D\ |\ x-y \in \C \text{ or } y-x \in \C \right\} \to (0,\infty).
\end{equation*}
\end{enumerate}

We think of a lattice as a node-spring system where the nodes are points in $\Z^D$ and $x,y \in \Z^D$ are connected by a spring iff $x-y \in \C$ or $y-x \in \C$, in which case, for brevity, we call $(x,y) \in \Z^D \times \Z^D$ a spring. The rest length of spring $(x,y)$ is given by $\L(x,y)$. Where $\L$ is irrelevant we speak of the lattice $(\Z^D,\C)$.

For $v \in \C$ we say that $v$ is nearest-neighbour if $\| v \|$ = 1 and next-nearest neighbour if $\| v \| = \sqrt{2}$. 
\end{definition}

(The reason for the restriction~\eqref{eq:local10} will become clear in Definition~\ref{def:lattice-order} below.) The simplest lattices are those that are homogeneous:

\begin{definition}[Homogeneous lattices]
A lattice is homogeneous if $\L$ is translation-invariant, that is, $\L(x,y)$ is a function only of $x-y$. In this case we view $\L$ as a function on $\C \cup -\C$ and write $\L(x-y)$ for $\L(x,y)$.
\end{definition}

A critical property of a lattice is its order:

\begin{definition}[Order of a lattice] \label{def:lattice-order}
The order of a lattice $(\Z^D,\C)$ is the smallest integer $K \in \N$ for which there exists a partition $\C_1, \dots, \C_K$ of $C$ such that each $\C_i$ is a set of linearly-independent vectors in $\R^D$. (Recall that a partition of a set $S$ is a collection of subsets $S_1,\dots,S_K$ of $S$ which are mutually disjoint and whose union is $S$.) Note that $K \geqslant \frac{|\C|}{D}$.
\end{definition}

The central object of our study is the continuum energy density of a lattice:

\begin{definition}[Continuum energy density of a lattice]
Let $(\Z^D,\C,\L)$ be a lattice. Then its continuum energy density $\Energy_{(\C,\L)} \colon \Def \to \R$ is given by
\begin{equation} \label{eq:continuum-energy}
\Energy_{(\C,\L)}(F)
 := \lim_{N \to \infty} \frac{1}{N^D}
     \inf_{ \substack{ u \colon \Z^D_N \to \R^D \\ u(x)|_{\partial \Z^D_N} = Fx } } \
     \sum_{ \substack{ x,y \in \Z^D_N \\ x-y \in \C } } 
     W \left( \frac{\| u(x) - u(y) \|}{\L(x,y)} \right),
\end{equation}
if it exists.
\end{definition}

Since they concern elasticity and not growth per se, we do not here entertain questions regarding the existence and form of $\Energy_{(\C,\L)}$ but refer the reader to~\cite{E:2007}.

Henceforth we assume that the lattices we consider possess continuum energy densities. In addition we assume that the homogeneous lattices  we consider satisfy the Cauchy-Born rule, see~\cite{E:2007}. It follows that the continuum energy density of a homogeneous lattice is of the form
\begin{equation} \label{eq:CB-energy}
\Energy_{(\C,\L)}(F)
 = \sum_{v \in \C} W \left( \frac{\| F v \|}{\L(v)} \right).
\end{equation}
We shall refer to such a continuum energy density as a \emph{Cauchy-Born continuum energy density}.

We also prefer to avoid questions that concern elastic homogenisation and not growth per se. To do this it is convenient to consider homogeneous representatives of inhomogeneous lattices:

\begin{definition}[Homogeneous representative of an inhomogeneous lattice] \label{def:homogeneous_representative}
Let $(\Z^D,\C,\L)$ be an (inhomogeneous) lattice and $(\Z^D,\C,\overline{\L})$ a homogeneous lattice. $(\Z^D,\C,\overline{\L})$ is a homogeneous representative of $(\Z^D,\C,\L)$ if their continuum energies agree, that is, if
\begin{equation*}
\Energy_{(\C,\L)}
 = \Energy_{(\C,\overline{\L})}.
\end{equation*}
(These energies are given by~\eqref{eq:continuum-energy} and~\eqref{eq:CB-energy} above.)
\end{definition}

The following Lemma can be understood as saying that the continuum energy density of a homogeneous lattice describes a solid only if its order is greater than one:

\begin{lemma}
The Cauchy-Born continuum energy density of a homogeneous lattice is shear-resisting only if the order of the lattice is greater than one.
\end{lemma}

\begin{proof}
Consider a homogeneous lattice $(\Z^D,\C)$ of order $1$ with $\C = \{ v_1,v_2,\dots,v_{|\C|} \}$. Then its Cauchy-Born continuum energy density, $\Energy$, is of the form
\begin{equation*}
\Energy(F) = \sum_{i=1}^{|\C|} W \left( \frac{\| F v_i\|}{L_i} \right)
\end{equation*}
for some $L_i$, $i=1,\dots,|\C|$. We show that $\Energy$ vanishes on a shear.

Since the lattice is of order 1, $|\C| \leqslant D$. If $|\C| < D$ we extend $\C$ to a basis $\tilde{\C} := \{ v_1,v_2,\dots,v_D \}$ of $\R^D$.

Pick $R_C \in SO(D)$ such that $\langle v_1, R_C v_2 \rangle \neq \langle v_1, v_2 \rangle$. Let $F_C \in \Def$ be the linear operator that satisfies
\begin{align*}
 F_C v_1 &= v_1 \\
 F_C v_i &= R_C v_i, \quad i=2,\dots,D.
\end{align*}
Note that $F_C \notin SO(D)$ since $\langle F_C v_1, F_C v_2 \rangle = \langle v_1, R_C v_2 \rangle \neq \langle v_1, v_2 \rangle$. On the other hand, clearly $\| \cdot \| \circ F = \| \cdot \|$ on $\tilde{\C}$. Thus $F_C$ is a shear and $\Energy(F_C) = \Energy(I)$. It follows that $\Energy$ is not shear resisting.
\end{proof}

%-------------
\subsection{An introductory example: The one-dimensional lattice $(\Z,\{e_1\})$}
\label{sec:1D}

We begin with a one-dimensional discrete system that consists of $N+1$ nodes located on $\Z_N$ with neighbouring nodes connected by springs of rest-length $L$, see Figure~\ref{fig:Z_1}. This is the initial (discrete) system. Now let the springs grow as follows: The rest-length of the spring connecting nodes located at $j-1$ and $j$ becomes $\left( g(\frac{j}{N}) - g(\frac{j-1}{N}) \right) L$ where $g \colon [0,1] \to \R_+$ is a specified increasing function. This is the grown (discrete) system.
\begin{figure}
\begin{center}
\includegraphics[scale=0.3]{1D-Spring-System-1.jpg}\hspace{-0.8cm}\includegraphics[scale=0.3]{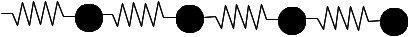}
\caption{The one-dimensional lattice $(\Z,\{e_1\})$.}
\label{fig:Z_1}
\end{center}
\end{figure}

When the displacement of the boundary of the system is specified, $u_0 = 0$ and $u_N = F N$, $F>0$, the energy of the systems is obtained by a minimisation over the displacements of the interior nodes:
\begin{align*} 
W_i^N(F)
 &:= \min_{\substack{u_j,\ j=1,\dots,N-1 \\ u_0 = 0 \\ u_N = F N}} \ \sum_{j=1}^N W \left( \frac{| u_j - u_{j-1} |}{L} \right), \\
W_g^N(F)
 &:= \min_{\substack{u_j,\ j=1,\dots,N-1 \\ u_0 = 0 \\ u_N = F N}} \ \sum_{j=1}^N \left( g(\frac{j}{N}) - g(\frac{j-1}{N}) \right)^p
  W \left( \frac{| u_j - u_{j-1} |}{ \left( g(\frac{j}{N}) - g(\frac{j-1}{N}) \right) L } \right)
\end{align*}
where $W$ is the energy density of the springs and the subscripts $i$ and $g$ denote the initial and grown systems.

It is an easy exercise to show that (with some coercivity conditions on $W$) these discrete energies converge to a continuum limit:
\begin{alignat*}{2}
W_i(F)
 &=  \lim_{N \to \infty} \frac{1}{N} W_i^N(F)
 &&= \min_{\substack{u \colon [0,1] \to \R \\ u(0)=0, u(1)=F}} \int_0^1 W \left( \frac{Du(x)}{L} \right)\ \d x \\
W_g(F)
 &= \lim_{N \to \infty} \frac{1}{N} W_i^N(F)
 &&= \min_{\substack{u \colon [0,1] \to \R \\ u(0)=0, u(1)=F}} \int_0^1 G^p(x) W \left( \frac{Du(x)}{L G(x)} \right)\ \d x \\
\end{alignat*}
where $G = g'$.

Thus the grown continuum elastic energy-density of the system $W_g$ is related to the initial continuum elastic energy-density $W_i$ through:
\begin{equation} \label{eq:F=AG}
W_g(\cdot)
 = G^p W_i(\cdot G^{-1})
\end{equation}
where $W_i \equiv W$. Thus when $p=0$ (that is, for recombination) this system satisfies the \emph{multiplicative-decomposition} first introduced in the context of morphoelasticity in~\cite{Rodriguez:1994}.

This illustrates, in the context of a simple example, the central question of this paper: How is the continuum elastic energy-density of the grown system related to that of the initial system? In particular is it related through quantities, such as $G$ in this example, that can be thought of as continuum descriptors of the growth of the system?

%--------
\section{Homogeneous lattices recombining homogeneously}
\label{sec:homogeneous_lattices_recombining_homogeneously}

%-------------
\subsection{The two-dimensional lattice $(\Z^2,\{e_1,e_2,e_{1\pm2} \})$}
\label{sec:Z^2_2-homogeneous}

Next we explore the simplest two-dimensional lattice that supports elasticity, $(\Z^2,\{e_1,e_2,e_{1\pm2} \})$, the two-dimensional square lattice with nodes located at $\Z^2$ and with nearest-neighbour and next-nearest-neighbour interactions, see Figure~\ref{fig:Z^2_2}.
\begin{figure}
\begin{center}
\includegraphics[scale=0.05]{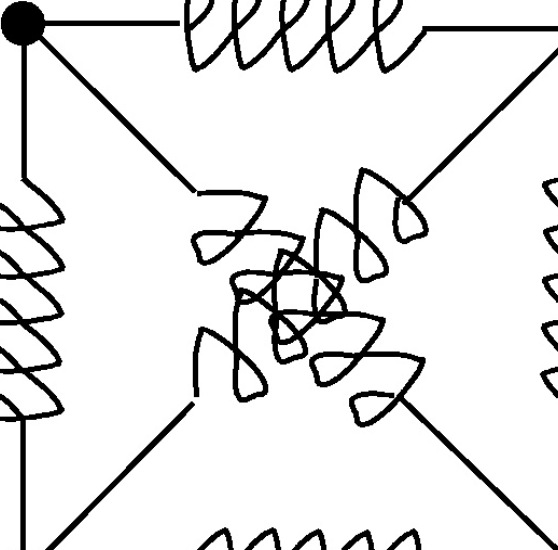}%
\includegraphics[scale=0.05]{Square-Spring-System-Truncated.jpg}% 
\includegraphics[scale=0.05]{Square-Spring-System-Truncated.jpg}%
\includegraphics[scale=0.05]{Square-Spring-System-Truncated.jpg} \\ \vspace{-0.04cm}
\includegraphics[scale=0.05]{Square-Spring-System-Truncated.jpg}%
\includegraphics[scale=0.05]{Square-Spring-System-Truncated.jpg}%
\includegraphics[scale=0.05]{Square-Spring-System-Truncated.jpg}%
\includegraphics[scale=0.05]{Square-Spring-System-Truncated.jpg} \\ \vspace{-0.04cm}
\includegraphics[scale=0.05]{Square-Spring-System-Truncated.jpg}%
\includegraphics[scale=0.05]{Square-Spring-System-Truncated.jpg}%
\includegraphics[scale=0.05]{Square-Spring-System-Truncated.jpg}%
\includegraphics[scale=0.05]{Square-Spring-System-Truncated.jpg} \\ \vspace{-0.04cm}
\includegraphics[scale=0.05]{Square-Spring-System-Truncated.jpg}%
\includegraphics[scale=0.05]{Square-Spring-System-Truncated.jpg}%
\includegraphics[scale=0.05]{Square-Spring-System-Truncated.jpg}%
\includegraphics[scale=0.05]{Square-Spring-System-Truncated.jpg}%
\caption{The two-dimensional lattice $\Z^2_2$}
\label{fig:Z^2_2}
\end{center}
\end{figure}

Let the initial (i.e., at some instant during growth) rest-lengths of the springs be $L_1$, $L_2$, $L_{1+2}$ and $L_{1-2}$ for the horizontal, vertical, right-diagonal and left-diagonal springs respectively; note that these are arbitrary but fixed. At some later instant let the growth (due to recombination) in the springs be $g_1$, $g_2$, $g_{1+2}$ and $g_{1-2}$ for the horizontal, vertical, right-diagonal and left-diagonal springs respectively. The following result relates the continuum limits of the initial and grown systems:

\begin{theorem}[Homogenenously-recombined $(\Z^2,\{e_1,e_2,e_{1\pm2} \})$] \label{thm:Z^2_2}
Let the initial and homogeneously-recombined $(\Z^2,\{e_1,e_2,e_{1\pm2} \})$ lattices have Cauchy-Born continuum energy-densities $W_i$ and $W_g$ respectively. Then:
\begin{enumerate}
\item \label{it:energy-deformation} There exist $W_1, W_2 \colon \R^{2 \times 2} \to \R$ (independent of $g_1$, $g_2$, $g_{1+2}$ and $g_{1-2}$) and $G_1,G_2 \in \R^{2 \times 2}$ (independent of $L_1$, $L_2$, $L_{1+2}$ and $L_{1-2}$) such that $\forall F \in \R^{2 \times 2}$,
\begin{subequations} \label{eq:energy-deformation}
\begin{align}
W_i(F) &= W_1(F) + W_2(F), \label{eq:energy-deformation-initial} \\
W_g(F) &= W_1(F G_1^{-1}) + W_2(F G_2^{-1}). \label{eq:energy-deformation-grown}
\end{align}
\end{subequations}
\item \label{it:F!=AG} If there exists $G \in \R^{2 \times 2}$ such that $\forall F \in \R^{2 \times 2}$,
\begin{equation} \label{eq:F!=AG}
W_g(F) = W_i(F G^{-1})
\end{equation}
then, in fact the lattice has dilated under recombination and $G$ can be picked to be a multiple of the identity.
\end{enumerate}
\end{theorem}

Clearly $G_1$ and $G_2$ describe the growth of the system at the continuum level and Item~\eqref{it:energy-deformation} introduces an \emph{energy-deformation decomposition} that relates the initial and grown systems. Item~\eqref{it:F!=AG} on the other hand shows that this is generically optimal and that a \emph{multipicative-decomposition of the deformation gradient} can be expected to hold only when the growth is isotropic, see Example~\ref{eg:Z_2} below. In the sequel when we refer to ``energy-deformation decomposition'' we mean a result of the form~\eqref{eq:energy-deformation} and when we refer to ``multiplicative decomposition'' we mean a result of the form~\eqref{eq:F!=AG}.

Before we present the formal proof we first explain the central insight (Remark~\ref{rem:Z^2_2-1}), present explicit expressions for $W_1$, $W_2$, $G_1$, $G_2$ (Lemma~\ref{lem:!unique_decomposition}) and correct what appears to be a popular misunderstanding in Remark~\ref{rem:Z^2_2-2}.

\begin{remark}[Insight behind the proof] \label{rem:Z^2_2-1}
Since the springs lie along $e_\alpha$, $\alpha \in \{ 1,2,1\pm2 \}$ (in the reference configuration), the elastic energy density of the lattice depends on a deformation $F$ (only) through the quantities $\| F e_\alpha \|$, $\alpha \in \{ 1,2,1\pm2 \}$. Had these vectors formed a basis for $\R^2$ the growth of these springs could have been represented by a linear operator on $\R^2$. But this generically not the case and thus, in general, no $G \in \R^{2 \times 2}$ can \emph{geometrically} represent the growth of the system. However from the four vectors we can form two sets of basis and thus the growth of the springs can be represented by two linear operators on $\R^2$, namely, $G_1, G_2 \in \R^{2 \times 2}$. Item \eqref{it:energy-deformation} then follows from an elementary calculation.

Suppose, however, that the growth $g_\alpha$, $\alpha \in \{1,2,1\pm2\}$ of the springs can in fact be represented by a linear operator $G \in \R^{2 \times 2}$, that is, suppose that there exists a $G \in \R^{2 \times 2}$ such that $g_\alpha = G e_\alpha$, $\alpha \in \{1,2,1\pm2\}$. Another elementary calculation shows that, nevertheless,~\eqref{eq:F!=AG} does not hold except when $G \parallel I$. The reason for this is that the energy density of a lattice with springs along $G e_\alpha$, $\alpha \in \{ 1,2,1\pm2 \}$ is not related to the energy density of a lattice with springs along $e_\alpha$, $\alpha \in \{ 1,2,1\pm2 \}$ through a linear change of variables except when $G \parallel I$.

Thus analogies to crystal plasticity (where the deformation gradient is decomposed into a plastic-part and an elastic-part) are misleading since crystal plasticity does not change the (global) metric on the lattice but growth does.
\end{remark}

The next Lemma shows that $W_1$ and $W_2$ can be picked to be purely elastic (i.e., independent of growth); and conversely $G_1$ and $G_2$ can be picked to be purely geometric (i.e., independent of elasticity):

\begin{lemma}[Non-uniqueness of decomposition] \label{lem:!unique_decomposition}
One choice of $(W_1, G_1)$ and $(W_2, G_2)$ is
\begin{subequations} \label{eq:energy-deformation1}
\begin{align}
W_1(F)
 &= W \left( \frac{\|F e_1 \|}{L_1} \right) + W \left( \frac{\|F e_2 \|}{L_2} \right),
 &G_1
 &= \begin{pmatrix} g_1 & 0 \\ 0 & g_2 \end{pmatrix},  \label{eq:energy-deformation1a} \\
W_2(F)
 &= W \left( \frac{\|F e_{1+2} \|}{L_{1+2}} \right) + W \left( \frac{\|F e_{1-2} \|}{L_{1-2}} \right),
 &G_2
 &= R_{\frac{\pi}{4}} \begin{pmatrix} g_{1+2} & 0 \\ 0 & g_{1-2} \end{pmatrix} R_{\frac{\pi}{4}}^T \notag \\
 &&&= \frac{1}{2} \begin{pmatrix} g_{1+2} + g_{1-2} & g_{1+2} - g_{1-2} \\ g_{1+2} - g_{1-2} & g_{1+2} + g_{1-2} \end{pmatrix}, \label{eq:energy-deformation1b}
\end{align}
\end{subequations}
where $W$ is the energy density of the springs.

Two other genuinely different (i.e., not obtainable by addition of constants) choices exist:
\begin{subequations} \label{eq:energy-deformation2}
\begin{align}
W_1(F)
 &= W \left( \frac{\|F e_1 \|}{L_1} \right) + W \left( \frac{\|F e_{1+2} \|}{L_{1+2}} \right),
 &G_1
 &= \begin{pmatrix} g_1 & g_{1+2} - g_1 \\ 0 & g_{1+2} \end{pmatrix}, \\
W_2(F)
 &=  \left( \frac{\|F e_2 \|}{L_2} \right) + W \left( \frac{\|F e_{1-2} \|}{L_{1-2}} \right),
 &G_2
 &= \begin{pmatrix} g_{1-2} & 0 \\ g_2 - g_{1-2} & g_2 \end{pmatrix}.
\end{align}
\end{subequations}
and
\begin{subequations} \label{eq:energy-deformation3}
\begin{align}
W_1(F)
 &= W \left( \frac{\|F e_1 \|}{L_1} \right) + W \left( \frac{\|F e_{1-2} \|}{L_{1-2}} \right),
 &G_1
 &= \begin{pmatrix} g_1 & g_1 - g_{1-2} \\ 0 & g_{1-2} \end{pmatrix}, \\
W_2(F)
 &= W \left( \frac{\|F e_2 \|}{L_2} \right) + W \left( \frac{\|F e_{1+2} \|}{L_{1+2}} \right),
 &G_2
 &= \begin{pmatrix} g_{1+2} & 0 \\ g_{1+2} - g_2 & g_2 \end{pmatrix}.
\end{align}
\end{subequations}
\end{lemma}

\begin{remark} \label{rem:shear_vanishing}
Each of the six energies introduced in Lemma~\ref{lem:!unique_decomposition}, in~\eqref{eq:energy-deformation1},\eqref{eq:energy-deformation2} and~\eqref{eq:energy-deformation3} vanishes on a one-parameter family of shears: For~\eqref{eq:energy-deformation1}, $W_1$ and $W_2$ vanish on
\begin{equation*}
SO(2) \begin{pmatrix} 1 & \cos \theta \\ 0 & \sin \theta \end{pmatrix}, \theta \in \R
 \qquad \text{and} \qquad
  SO(2) \begin{pmatrix} \frac{1 - \cos \theta}{\sqrt{2}} & \frac{1 + \cos \theta}{\sqrt{2}} \\ -\frac{\sin \theta}{\sqrt{2}} & \frac{\sin \theta}{\sqrt{2}} \end{pmatrix}, \theta \in \R
\end{equation*}
respectively; for~\eqref{eq:energy-deformation2}, $W_1$ and $W_2$ vanish on
\begin{equation*}
SO(2) \begin{pmatrix} 1 & \sqrt{2} \cos \theta - 1 \\ 0 & \sqrt{2} \sin \theta \end{pmatrix}, \theta \in \R
 \qquad \text{and} \qquad
  SO(2) \begin{pmatrix} \sqrt{2} \cos \theta + 1 & 1 \\ \sqrt{2} \sin \theta & 0 \end{pmatrix}, \theta \in \R
\end{equation*}
respectively; and for~\eqref{eq:energy-deformation3}, $W_1$ and $W_2$ vanish on
\begin{equation*}
SO(2) \begin{pmatrix} 1 & \sqrt{2} \cos \theta + 1 \\ 0 & \sqrt{2} \sin \theta \end{pmatrix}, \theta \in \R
 \qquad \text{and} \qquad
  SO(2) \begin{pmatrix} \sqrt{2} \cos \theta - 1 & 1 \\ \sqrt{2} \sin \theta & 0 \end{pmatrix}, \theta \in \R
\end{equation*}
respectively.
\end{remark}

\begin{remark}[Residual-stress and shear-resistance] \label{rem:Z^2_2-2}
Note that Theorem~\ref{thm:Z^2_2} holds regardless of whether the system, either initially or when grown, is residually stressed or not. In particular~\eqref{eq:F!=AG} holds for dialatory growth even in the presence of residual stress and does not hold for non-dialatory growth even in the absence of residual-stress. This shows that, contrary to what appears to be popular opinion, the essential issue is not the build-up of residual stress under growth. Rather it is the interplay between shear-resistance and growth. Indeed the multiplicative-decomposition is valid for a homogeneously-growing fluid as Lemma~\ref{lem:Z^n_1} below illustrates. (See also Example~\ref{eg:Z_2} and~\S\ref{sec:Z^2_2-2} below.)
\end{remark}

We are now ready to prove Theorem~\ref{thm:Z^2_2} and Lemma~\ref{lem:!unique_decomposition}:

\begin{proof}
The Cauchy-Born continuum energy densities are
\begin{subequations}
\begin{align}
W_i(F)
 &= \sum_{ \alpha \in \{1,2,1+2,1-2\} } W \left( \frac{\| F e_\alpha \|}{L_\alpha} \right), \label{eq:local1a} \\
W_g(F)
 &= \sum_{ \alpha \in \{1,2,1+2,1-2\} } W \left( \frac{\| F e_\alpha \|}{L_\alpha} g_\alpha^{-1} \right). \label{eq:local1b}
\end{align}
\end{subequations}
Now let $W_1$, $W_2$, $G_1$ and $G_2$ be as in~\eqref{eq:energy-deformation1}. From~\eqref{eq:local1a} it is immediate that~\eqref{eq:energy-deformation-initial} is satisfied. Thus, to prove Item~\eqref{it:energy-deformation} it suffices to verify~\eqref{eq:energy-deformation-grown}: For $F \in \R^{2 \times 2}$:

Since $G_1^{-1} e_\alpha = g_\alpha^{-1} e_\alpha$, $\alpha \in \{1,2\}$,
\begin{subequations} \label{eq:local2}
\begin{align}
W_1(F G_1^{-1})
 &= W \left( \frac{\| F G_1^{-1} e_1 \|}{L_1} \right) + W \left( \frac{\| F G_1^{-1} e_2 \|}{L_2} \right) \notag \\
 &= W \left( \frac{\| F e_1 \|}{L_1} g_1^{-1} \right) + W \left( \frac{\| F e_2 \|}{L_2} g_2^{-1} \right).
\end{align}
Similarly, since $G_2^{-1} e_\alpha = g_\alpha^{-1} e_\alpha$, $\alpha \in \{ 1\pm 2 \}$,
\begin{align}
W_2(F G_2^{-1})
 &= W \left( \frac{\| F G_2^{-1} e_{1+2} \|}{L_{1+2}} \right) + W \left( \frac{\| F G_2^{-1} e_{1-2} \|}{L_{1-2}} \right) \notag \\
 &= W \left( \frac{\| F e_{1+2} \|}{L_{1+2}} g_{1+2}^{-1} \right) + W \left( \frac{\| F e_{1-2} \|}{L_{1-2}} g_{1-2}^{-1} \right).
\end{align}
\end{subequations}
Equations~\eqref{eq:local1b} and~\eqref{eq:local2} immediately yield~\eqref{eq:energy-deformation-grown}. 

Similarly, when $W_1$, $W_2$, $G_1$ and $G_2$ are as in~\eqref{eq:energy-deformation2}, \eqref{eq:energy-deformation-initial} is immediate and to verify~\eqref{eq:energy-deformation-grown} it suffices to verify that
\begin{alignat*}{2}
G_1^{-1} e_\alpha
 &= g_\alpha^{-1} e_\alpha, \quad &&\alpha \in \{1,1+2\}, \\
G_2^{-1} e_\alpha
 &= g_\alpha^{-1} e_\alpha, \quad &&\alpha \in \{2,1-2\}.
\end{alignat*}
This in turn immediately follows from
\begin{align*}
G_1^{-1}
 &= \begin{pmatrix} g_1^{-1} & g_{1+2}^{-1} - g_1^{-1} \\ 0 & g_{1+2}^{-1} \end{pmatrix}, \\
G_2^{-1} 
 &= \begin{pmatrix} g_{1-2}^{-1} & 0 \\ g_2^{-1} - g_{1-2}^{-1} & g_2^{-1} \end{pmatrix}.
\end{align*}

Finally, when $W_1$, $W_2$, $G_1$ and $G_2$ are as in~\eqref{eq:energy-deformation3}, \eqref{eq:energy-deformation-initial} is again immediate and to verify~\eqref{eq:energy-deformation-grown} its suffices to verify that
\begin{alignat*}{2}
G_1^{-1} e_\alpha
 &= g_\alpha^{-1} e_\alpha, \quad &&\alpha \in \{1,1-2\}, \\
G_2^{-1} e_\alpha
 &= g_\alpha^{-1} e_\alpha, \quad &&\alpha \in \{2,1+2\}.
\end{alignat*}
This in turn immediately follows from
\begin{align*}
G_1^{-1}
 &= \begin{pmatrix} g_1^{-1} & g_1^{-1} - g_{1-2}^{-1} \\ 0 & g_{1-2}^{-1} \end{pmatrix}, \\
G_2^{-1} 
 &= \begin{pmatrix} g_{1+2}^{-1} & 0 \\ g_{1+2}^{-1} - g_2^{-1} & g_2^{-1} \end{pmatrix}.
\end{align*}

We now turn to Item~\eqref{it:F!=AG}. Suppose there exists $G \in \R^{2 \times 2}$ such that~\eqref{eq:F!=AG} holds $\forall F \in \R^{2 \times 2}$. That is, $\forall F \in \R^{2 \times 2}$,
\begin{equation*}
W_g(F)
 = \sum_{ \alpha \in \{1,2,1+2,1-2\} } W \left( \frac{\| F e_\alpha \|}{L_\alpha} g_\alpha^{-1} \right)
 = \sum_{ \alpha \in \{1,2,1+2,1-2\} } W \left( \frac{\| F G^{-1} e_\alpha \|}{L_\alpha} \right)
 = W_i(F G^{-1}).
\end{equation*}
Recall that $G$ is independent of $L_1$, $L_2$, $L_{1+2}$ and $L_{1-2}$. Since each $L_\alpha$, $\alpha \in \{ 1, 2, 1 \pm 2 \}$ is arbitrary (though fixed) and appears only in one term it follows that the equality holds term-by-term: $\forall F \in \R^{2 \times 2}$,
\begin{equation} \label{eq:local3}
\| F G^{-1} e_\alpha \| = \| F e_\alpha \| g_\alpha^{-1}, \qquad \alpha \in \{ 1, 2, 1 \pm 2 \}
\end{equation}
With the choice $F=I$, setting $H=F G^{-1}$, this implies that
\begin{alignat*}{3}
\frac{2}{g_{1+2}^2}
 &= \| H (e_1 + e_2) \|^2
 &&= \| H e_1 \|^2 + 2 \langle H e_1, H e_2 \rangle + \| H e_2 \|^2
 &&= \frac{1}{g_1^2} + 2 \langle H e_1, H e_2 \rangle + \frac{1}{g_2^2}, \\
\frac{2}{g_{1-2}^2}
 &= \| H (e_1 - e_2) \|^2
 &&= \| H e_1 \|^2 - 2 \langle H e_1, H e_2 \rangle + \| H e_2 \|^2
 &&= \frac{1}{g_1^2} - 2 \langle H e_1, H e_2 \rangle + \frac{1}{g_2^2}.
\end{alignat*}
Thus,
\begin{subequations} \label{eq:local4}
\begin{equation}
\frac{1}{g_1^2} + \frac{1}{g_2^2}
 = \frac{1}{g_{1+2}^2} + \frac{1}{g_{1-2}^2}.
 \end{equation}
With the choice $F=G$, \eqref{eq:local3} implies that
\begin{alignat*}{3}
2 g_{1+2}^2
 &= \| G (e_1 + e_2) \|^2
 &&= \| G e_1 \|^2 + 2 \langle G e_1, G e_2 \rangle + \| G e_2 \|^2
 &&= g_1^2 + 2 \langle G e_1, G e_2 \rangle + g_2^2, \\
2 g_{1-2}^2
 &= \| G (e_1 - e_2) \|^2
 &&= \| G e_1 \|^2 - 2 \langle G e_1, G e_2 \rangle + \| G e_2 \|^2
 &&= g_1^2 - 2 \langle G e_1, G e_2 \rangle + g_2^2.
\end{alignat*}
Thus,
\begin{equation}
g_1^2 + g_2^2
 = g_{1+2}^2 + g_{1-2}^2.
 \end{equation}
With the choice $F= \left( \begin{smallmatrix} g_1 & 0 \\ 0 & g_2 \end{smallmatrix} \right)$, setting $H=F G^{-1}$, \eqref{eq:local3} implies,
\begin{alignat*}{3}
\frac{g_1^2 + g_2^2}{g_{1+2}^2}
 &= \| H (e_1 + e_2 ) \|^2
 &&= \| H e_1 \|^2 + 2 \langle H e_1, H e_2 \rangle + \| H e_2 \|^2
 &&= 2 + 2 \langle H e_1, H e_2 \rangle, \\
\frac{g_1^2 + g_2^2}{g_{1-2}^2}
 &= \| H (e_1 - e_2 ) \|^2
 &&= \| H e_1 \|^2 - 2 \langle H e_1, H e_2 \rangle + \| H e_2 \|^2
 &&= 2 - 2 \langle H e_1, H e_2 \rangle.
\end{alignat*}
Thus,
\begin{equation}
\left( g_1^2 + g_2^2 \right) \left( \frac{1}{g_{1+2}^2} + \frac{1}{g_{1-2}^2} \right) = 4.
\end{equation}
\end{subequations}
An easy calculation shows that~\eqref{eq:local4} implies that $g_1 = g_2 = g_{1+2} = g_{1-2} =: g$ (say). In other words, the lattice has dilated by $g$ and, from~\eqref{eq:local3} we can pick $G = g I$.
\end{proof}

\begin{remark} \label{rem:Z^2_2-3}
The proof shows why decomposition into two energies suffices: As Remark~\ref{rem:shear_vanishing} shows each of $W_1$ and $W_2$ in~\eqref{eq:energy-deformation} vanishes on a shear. This sufficiently decouples the energy that each ``growth mode'' does not influence the energy of the other ``deformation mode'', i.e., the deformation that the other energy function can see. Thus $G_1$ does not influence $W_2$ and vice versa.

On the other hand that $W_1(\cdot G_1^{-1})$ and $W_2(\cdot G_2^{-1})$ are evaluated at the same $F$ couples these energies together. Indeed the theorem proves that this provides all the coupling that is needed.
\end{remark}

Remarks~\ref{rem:Z^2_2-1}, \ref{rem:Z^2_2-2} and~\ref{rem:Z^2_2-3} can be illustrated by a one-dimensional example:

\begin{example} \label{eg:Z_2}
Here we present a (geometrically) one-dimensional example where the multiplicative-decomposition fails but an energy-deformation decomposition continues to hold. Consider a one-dimensional discrete system where both the nearest neighbours and the next-nearest neighbours are connected by springs. This is the initial system.

Let the nearest-neighbour springs grow by $g_1$ and the next-nearest neighbour-springs grow by $g_2$. This is the grown system.

Then, as is easy to see, the continuum energy densities are related by
\begin{align*}
W_i(\cdot)
 &= \frac{1}{2} W_1(\cdot) + \frac{1}{2} W_2(\cdot), \\ 
W_g(\cdot)
 &= \frac{1}{2} W_1(\cdot g_1^{-1}) + \frac{1}{2} W_2(\cdot g_2^{-1})
\end{align*}
where $W_1$ and $W_2$ are the energies of the nearest-neighbour and next-nearest neighbour springs, respectively.

But, when $g_1 \neq g_2$, there is no $g$ such that $W_g(\cdot) = W_i(\cdot g^{-1})$. This is, of course, due to the interaction between the nearest-neighbour springs and the next-nearest-neighbour springs. In dimensions larger than one, when shear resistance is present, such interaction is generic; indeed such interaction is the source of shear resistance.

On the other hand, $g_1 = g_2$, is analogous to isotropic growth.
\end{example}

%-----
\subsubsection{Comparison of the two decompositions for the lattice $(\Z^2,\{e_1,e_2,e_{1\pm2} \})$}
\label{sec:Z^2_2-2}

To illustrate the difference between the energy-deformation decomposition~\eqref{eq:energy-deformation} and the multiplicative decomposition~\eqref{eq:F!=AG} we consider the lattice $(\Z^2,\{e_1,e_2,e_{1\pm2} \})$ with initial rest-lengths of the springs $L_1 = L_2 = 1$ and $L_{1+2} = L_{1-2} = \sqrt{2}$. Note that this is a residually unstressed configuration.

In Examples~\ref{eg:non-diagonal_shrinking} to~\ref{eg:residually-stressed_sheared_lattice} below we compute the fractional error
\begin{equation} \label{eq:local9}
\frac{W_i(F G^{-1})}{ W_g (F)} -1
\end{equation}
that results when the energy of the recombined system is computed using the multiplicative decomposition for 
\begin{equation} \label{eq:local5}
F \in \left\{ \begin{pmatrix} \l_1 & \l_3 \\ 0 & \l_2 \end{pmatrix} \middle| \ \l_1, \l_2 \in [0.8, 1.25], \l_3 \in [-0.5,0.5] \right\}.
\end{equation}
To do this, we should first associate a (single) $G \in \R^{2 \times 2}$ with the grown system. It is natural to identify $G$ with the deformation gradient corresponding to the ground state of the grown system. This specifies $G$ unto pre-multiplication by a rotation in $SO(2)$. In keeping with the form of $F$ in~\eqref{eq:local5} we pick $G$ to be of the form
\begin{equation*}
G = \begin{pmatrix} G_{11} & G_{12} \\ 0 & G_{22} \end{pmatrix}
\end{equation*}
and require $G_{11}, G_{22} > 0$.

\begin{example}[Recombination that shrinks only the non-diagonal springs by 10\%] \label{eg:non-diagonal_shrinking} 
Let the recombination in the system be $g_1 = g_2 = 1$, $g_{1+2} = g_{1-2} = 0.9$. An elementary calculation shows that $G= 0.94475 I$. (Note however that the growth is not isotropic.) The region where the fractional error~\eqref{eq:local9} is greater than 10\% is shown in Figure~\ref{fig:non-diagonal_shrinking}.
\end{example}

\begin{figure}
\begin{center}
\begin{subfigure}{0.45\textwidth}
\includegraphics[width=\textwidth]{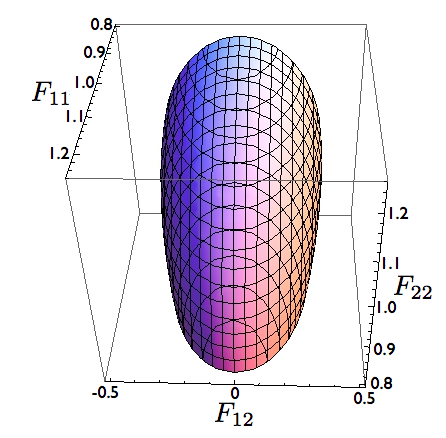}
\caption{}
\label{fig:non-diagonal_shrinking}
\end{subfigure}
\begin{subfigure}{0.45\textwidth}
\includegraphics[width=\textwidth]{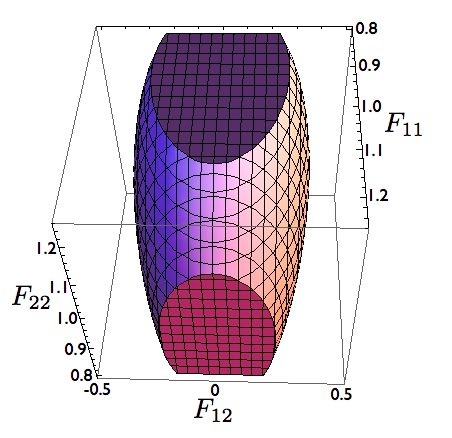}
\caption{}
\label{fig:diagonal_elongation}
\end{subfigure}
\begin{subfigure}{0.45\textwidth}
\includegraphics[width=\textwidth]{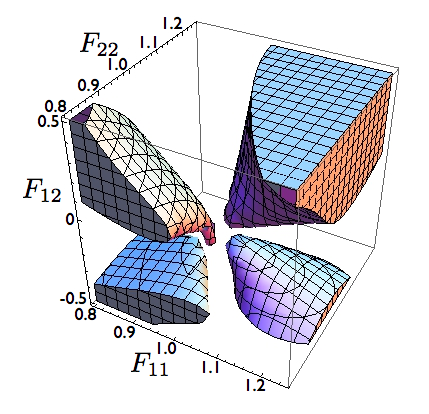}
\caption{}
\label{fig:residually-unstressed_sheared_lattice}
\end{subfigure}
\begin{subfigure}{0.45\textwidth}
\includegraphics[width=\textwidth]{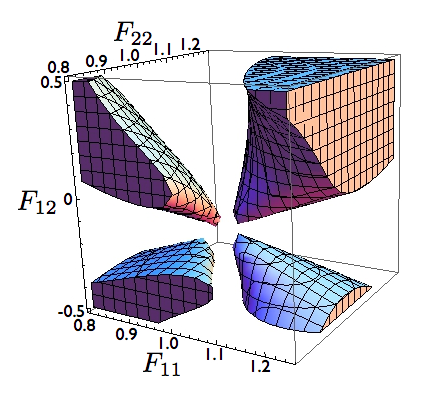}
\caption{}
\label{fig:residually-stressed_sheared_lattice}
\end{subfigure}
\caption{Regions where the fractional error in energy is greater than 10\% for (\subref{fig:non-diagonal_shrinking}) Example~\ref{eg:non-diagonal_shrinking}, (\subref{fig:diagonal_elongation}) Example~\ref{eg:diagonal_elongation}, (\subref{fig:residually-unstressed_sheared_lattice}) Example~\ref{eg:residually-unstressed_sheared_lattice} and (\subref{fig:residually-stressed_sheared_lattice}) Example~\ref{eg:residually-stressed_sheared_lattice}. Note that (\subref{fig:residually-unstressed_sheared_lattice}) is very similar to (\subref{fig:residually-stressed_sheared_lattice}), showing the unimportance of residual stress; on the other hand these regions are very different from (\subref{fig:non-diagonal_shrinking}) and (\subref{fig:residually-unstressed_sheared_lattice}) showing the importance of shear caused by growth. This is in keeping with Remark~\ref{rem:Z^2_2-2}.}
%\label{}
\end{center}
\end{figure}

\begin{example}[Recombination that elongates only the diagonal springs by 10\%] \label{eg:diagonal_elongation} 
Let the recombination in the system be $g_1 = g_2 = 1$, $g_{1+2} = g_{1-2} = 1.1$. An elementary calculation shows that $G= 1.04525 I$. (Note however that the growth is not isotropic.) The regions where the fractional error~\eqref{eq:local9} is greater than 10\% is shown in Figure~\ref{fig:diagonal_elongation}.
\end{example}

\begin{example}[Recombination of about 10\% resulting in a residually-unstressed sheared lattice] \label{eg:residually-unstressed_sheared_lattice} 
Let the recombination in the system be $g_1 = g_2 = 1$, $g_{1+2} = 0.9$ and $g_{1-2} = \sqrt{2 - 0.9^2} = 1.09087$. Note that 
\begin{equation*}
\left( \sqrt{2} g_{1+2} \right)^2 + \left( \sqrt{2} g_{1-2} \right)^2
 = 2 (g_1^2 + g_2^2),
\end{equation*}
so the recombined system is also residually unstressed. An elementary calculation shows that
\begin{equation*}
G = \begin{pmatrix} 1 & -0.19 \\ 0 & 0.98178 \end{pmatrix}.
\end{equation*}
The region where the fractional error~\eqref{eq:local9} is greater than 10\% is shown in Figure~\ref{fig:residually-unstressed_sheared_lattice}.
\end{example}

\begin{example}[Recombination resulting in a residually-stressed sheared lattice] \label{eg:residually-stressed_sheared_lattice} 
Let the recombination in the system be $g_1 = g_2 = 1$, $g_{1+2} = 0.9$ and $g_{1-2} = 1.1$. Note that 
\begin{equation*}
\left( \sqrt{2} g_{1+2} \right)^2 + \left( \sqrt{2} g_{1-2} \right)^2
 \neq 2 (g_1^2 + g_2^2),
\end{equation*}
so the recombined system is residually stressed. An elementary calculation shows that
\begin{equation*}
G = \begin{pmatrix} 1.00243 & -0.19757 \\ 0 & 0.98277 \end{pmatrix}.
\end{equation*}
The region where the fractional error~\eqref{eq:local9} greater than 10\% and 20\% are shown in Figure~\ref{fig:residually-stressed_sheared_lattice}.
\end{example}

%-------------
\subsection{Lattices on $\Z^n$.}

Next we turn to generic homogeneous lattices undergoing homogeneous growth. 

Consider the lattice $(\Z^D,\C)$. Let the initial rest-lengths of the springs be $L_v$, $v \in \C$. Let the growth (due to recombination) in the springs be $g_v$, $v \in \C$. The following result relates the continuum limits of the initial and grown systems:

\begin{theorem}[Homogenenously-recombined $(\Z^D,\C)$] \label{thm:Z^D}
Let $(\Z^D,\C)$ be a lattice of order $K$ and let its Cauchy-Born continuum energy density be $W_i$ and $W_g$, before and after homogeneous recombination, respectively. Then there exist $W_1, \dots, W_K \colon \R^{D \times D} \to \R$ (independent of $g_\alpha$, $\alpha \in \C$) and $G_1, \dots, G_K \in \R^{D \times D}$ (independent of $L_\alpha$, $\alpha \in \C$) such that $\forall F \in \R^{D \times D}$,
\begin{subequations} \label{eq:energy-deformation_(generic)}
\begin{align}
W_i(F) &= \sum_{k=1}^K W_k(F), \label{eq:energy-deformation-initial_(generic)} \\
W_g(F) &= \sum_{k=1}^K W_k(F G_k^{-1}). \label{eq:energy-deformation-grown_(generic)}
\end{align}
\end{subequations}
\end{theorem}

When $K=1$, from Theorem~\ref{thm:Z^D} the energy-deformation decomposition reduces to a multiplicative decomposition of the deformation:

\begin{corollary}[Multiplicative decomposition of the deformation holds for a growing fluid] \label{lem:Z^n_1}
Let the assumptions of Theorem~\ref{thm:Z^D} hold for a lattice of order 1. Then there exists $G \in \Def$ such that $\forall F \in \Def$,
\begin{equation*}
W_g(F)
 = W_i(F G^{-1}).
\end{equation*}
\end{corollary}

\begin{proof}[Proof of Theorem~\ref{thm:Z^D}]

Let $\C_k$, $k=1,\dots,K$ be a partition of $\C$ as specified in Definition~\ref{def:lattice-order}. From~\eqref{eq:CB-energy}, the Cauchy-Born continuum energy densities are of the form
\begin{subequations} \label{eq:local6}
\begin{alignat}{2}
W_i(F)
 &= \sum_{ v \in \C } W \left( \frac{\| F v \|}{L_v} \right)
 &&= \sum_{k=1}^K \sum_{ v \in \C_k } W \left( \frac{\| F v \|}{L_v} \right), \label{eq:local6a} \\
W_g(F)
 &= \sum_{ v \in \C } W \left( \frac{\| F v \|}{L_v} g_v^{-1} \right)
 &&= \sum_{k=1}^K \sum_{ v \in \C_k } W \left( \frac{\| F v \|}{L_v} g_v^{-1} \right). \label{eq:local6b}
\end{alignat}
\end{subequations}
Let $\tilde{\C}_k \supseteq \C_k$ be a basis of $\R^D$. We define $G_k$ by
\begin{subequations} \label{eq:local7}
\begin{alignat}{2}
G_k^{-1} v
 &= g_v^{-1} v, \quad && v \in \C_k, \\
G_k^{-1} v
 &= v, \quad && v \in \tilde{\C}_k \setminus \C_k. 
\end{alignat}
\end{subequations} 
and $W_k \colon \Def \to \R$ by
\begin{equation} \label{eq:local8} 
W_k(F) 
 := \sum_{ v \in \C_k } W \left( \frac{\| F v \|}{L_v} \right).
\end{equation}
Theorem~\ref{thm:Z^D} follows immediately from~\eqref{eq:local6}, \eqref{eq:local7} and~\eqref{eq:local8}.

\end{proof}

\begin{remark}
When $D=2$ and $\C = \{e_1,e_2,e_{1\pm2}\}$ we recover Theorem~\ref{thm:Z^2_2}\eqref{it:energy-deformation}. It is clear that $K=2$ and the possible partitions of $\{e_1,e_2,e_{1\pm2}\}$ are
\begin{gather*}
\{ e_1, e_2 \}, \{ e_{1 \pm 2} \}, \\
\{ e_1, e_{1+2} \}, \{ e_2, e_{1-2} \}, \\
\{ e_1, e_{1-2} \}, \{ e_2, e_{1+2} \}
\end{gather*}
and the corresponding decompositions of the energy are given by~\eqref{eq:energy-deformation1},\eqref{eq:energy-deformation2} and~\eqref{eq:energy-deformation3}, respectively. 
\end{remark}

%--------
\section{Growth and homogenisation}
\label{sec:homogenisation}

Next we explore how the energy-deformation decomposition interacts with homogenisation by considering an inhomogeneous lattice subjected to inhomogeneous recombination. The natural question is whether Theorem~\ref{thm:Z^D} extends to this situation, that is, are the continuum energies of the initial and recombined lattice related by an energy-deformation decomposition?

Corollary~\ref{cor:Z^D} (of Theorem~\ref{thm:Z^D} below shows that indeed they are. We state it for inhomogeneous lattices with homogeneous representatives (Definition~\ref{def:homogeneous_representative}) because we do not wish to be distracted by the technicalities of homogenisation; a more general result is proved in~\cite{Chenchiah-Shipman3}.

\begin{corollary}[Recombination of inhomogenenous $(\Z^D,\C)$] \label{cor:Z^D}
Let the initial and recombined (inhomogeneous) lattices $(\Z^D,\C)$ lattices have homogeneous representatives whose Cauchy-Born continuum energy-densities are $W_i$ and $W_g$ respectively.  Then there exist $W_1,\dots, W_K \colon \Def \to \R$ and $\overline{G}_1, \dots, \overline{G}_K \in \Def$ such that $\forall F \in \Def$,
\begin{align*}
W_i(F) &= \sum_{k=1}^K W_k(F),\\
W_g(F) &= \sum_{k=1}^K W_k(F \overline{G}_i^{-1}).
\end{align*}
\end{corollary}

\begin{proof}
By assumption (see Definition~\ref{def:homogeneous_representative} and~\eqref{eq:CB-energy}) there exist $\L_i, \L_g \colon \C \cup -\C \to (0,\infty)$ such that $W_i$ and $W_g$ are the continuum energy densities of the homogeneous lattices $(\Z^D,\C,\L_i)$ and $(\Z^D,\C,\L_g)$, respectively. Thus we are in the situation of Theorem~\ref{thm:Z^D} and the result follows. %from item~\eqref{it:energy-deformation_(generic)} there.
\end{proof}

In the rest of the section we present computational simulations that illustrate Corollary~\ref{cor:Z^D} for the lattice $(\Z^2_N,\{e_1,e_2,e_{1\pm2} \})$ by demonstrating the existence of $W_1, W_2 \colon \R^{2 \times 2} \to \R$ and $\overline{G}_1,\overline{G}_2 \in \R^{2 \times 2}$ such that $\forall F \in \R^{2 \times 2}$,
\begin{align*}
W_i(F) &= W_1(F) + W_2(F), \\
W_g(F) &= W_1(F \overline{G}_1^{-1}) + W_2(F \overline{G}_2^{-1})
\end{align*}
(see~\eqref{eq:energy-deformation}). Here $W_i$ and $W_g$ are the continuum energies of the initial and rehomogenised lattices.

First, in Section~\ref{sec:homogeneous_lattices_recombining_inhomogeneously}, we consider lattices that are homogeneous prior to (but not after) recombination, and demonstrate that the energy-deformation decomposition is stable under homogenisation of growth. Then, in Section~\ref{sec:inhomogeneous_lattices}, we consider lattices that are inhomogeneous both before and after recombination and show that the energy-deformation decomposition is stable also under homogenisation of elasticity.

%-----
\subsection{Homogeneous $(\Z^2_N,\{e_1,e_2,e_{1\pm2} \})$ recombining inhomogeneously}
\label{sec:homogeneous_lattices_recombining_inhomogeneously}

\paragraph{Deformation gradients}
In the simulations below both the initial and the recombined system are subjected to affine displacement at the boundary corresponding to a deformation gradient $F$ (which is thus also the average deformation gradient in the system). This is picked from a set of the form
\begin{equation} \label{eq:F-domain}
\F
 := \left\{ \begin{pmatrix} \l_1 & \l_3 \\ 0 & \l_2 \end{pmatrix} \middle| \ \l_1, \l_2 \in [\frac{1}{\Lambda_d}, \Lambda_d], \l_3 \in [-\Lambda_s,\Lambda_s] \right\}
\end{equation}
where $\Lambda_d, \Lambda_s > 1$. 
However, to permit two-dimensional plots the figures below are plotted for deformations of the form
\begin{subequations} \label{eq:deformations}
\begin{align}
F^h(\l)
 &= \begin{pmatrix} \l & 0 \\ 0 & 1 \end{pmatrix}, \quad \l \in [\frac{1}{\Lambda_h}, \Lambda_h], \\
F^v(\l)
 &= \begin{pmatrix} 1 & 0 \\ 0 & \l \end{pmatrix}, \quad \l \in [\frac{1}{\Lambda_v}, \Lambda_v], \\
F^d(\l)
 &= \begin{pmatrix} \l & 0 \\ 0 & \l \end{pmatrix}, \quad \l \in [\frac{1}{\Lambda_d}, \Lambda_d], \\
F^s(\l)
 &= \begin{pmatrix} 1 & \l \\ 0 & 1 \end{pmatrix}, \quad \l \in [-\Lambda_s,\Lambda_s].
\end{align}
\end{subequations}
(The superscripts stand for ``horizontal'', ``vertical'', ``dilational'' and ``shear'' respectively.)

\paragraph{Existence of $W_1, W_2 \colon \R^{2 \times 2} \to \R$ and $\overline{G}_1,\overline{G}_2 \in \R^{2 \times 2}$}
$W_1$ and $W_2$ can be immediately picked from Lemma~\ref{lem:!unique_decomposition}. For $\overline{G}_1,\overline{G}_2 \in \R^{2 \times 2}$, we search for those that minimises the relative mean square error
\begin{equation*}
 \sum_{F \in \F} \left( \frac{W_g(F) - \left( W_1(F G_1^{-1}) + W_2(F G_2^{-1}) \right)}{W_g(F)} \right)^2
\end{equation*}
(where the sum is approximated by a uniform sampling over $\F$). We accept these to be the desired $\overline{G}_1,\overline{G}_2$ provided
\begin{enumerate}
\item The error is less than a tolerance, and
\item $\overline{G}_1$ and $\overline{G}_2$ are convergent, and the error is non-increasing, when the grid size increases. 
\end{enumerate}

We also compare the resulting $\overline{G}_1,\overline{G}_2$ with the best choices, for each $\l$, for the deformations in~\eqref{eq:deformations}. For example, $\overline{G}^s_1(\l),\overline{G}^s_2(\l)$ minimises
\begin{equation} \label{mse}
 \sum_{\l \in [-\Lambda_s,\Lambda_s]} \left( \frac{ W_g(F^s(\l)) - \left( W_1(F^s(\l) G_1^{-1}) + W_2(F^s(\l) G_2^{-1}) \right) }{W_g(F^s(\l))} \right)^2,
\end{equation}
and we compare it with $\overline{G}_1,\overline{G}_2$.

For the spring energy $W$ (see~\eqref{eq:recombination-energy}) we pick
\begin{equation} \label{Wqpick}
W(\cdot) = (\cdot - 1)^q
\end{equation}
with $q=2, 3, 4$, depending on the simulation.
 
\begin{simulation}[$(\Z^2_{16},\{e_1,e_2,e_{1\pm2} \})$ with alternating growth in diagonal springs] \label{sim:Baskin}  \par \indent

\textbf{Initial state:} $L_1 = L_2 = 1$ and $L_{1+2} = L_{1-2} = \sqrt{2}$. Note that the rest-state has zero energy, that is, is residual-stress free.

\textbf{Recombination}: $g_1 = g_2 = 1$, that is, the horizontal and vertical springs did not grow. $g_{1+2}$ was chosen to alternate between $1.2$ and $\sqrt{2-1.2^2}$ in a checkerboard fashion, and $g_{1-2}$ was chosen to be $\sqrt{2-g_{1-2}^2}$. Note that each unit cell of the recombined lattice possesses a zero-energy state. 

\textbf{Result}: We choose $W_1$ and $W_2$ as in~\eqref{eq:energy-deformation1}, with $q=2$ in \eqref{Wqpick}. Since $g_1 = g_2 = 1$, we  choose $\overline{G}_1$  to be of the form $\bar{\g}_1I$. From~\eqref{eq:energy-deformation1b} $G_2$ is of the form
\begin{equation*}
R_{\frac{\pi}{4}} \begin{pmatrix} \bar{\g}_{1+2} & 0 \\ 0 & \bar{\g}_{1-2} \end{pmatrix} R_{\frac{\pi}{4}}^T
\end{equation*}
Since the system possesses cubic symmetry we expect $\bar{\g}_{1+2} = \bar{\g}_{1-2}$, but we do not a priori impose this.

Let us first impose a deformation $F^d(\l)$ for  $\l \in [1/1.25,1.25]$: $\bar{\g}_1^d = 1.2010$,  $\bar{\g}_{1+2}^d= \bar{\g}_{1-2}^d =  0.8300$  minimises the least-square relative error over all $\l$ in the chosen range.  Figure~\ref{ggssDA120920Fdsv4}(a) compares the true energy with the energy given by homogenisation for each $\l$, and the fractional error of the homogenized energy to the true energy as a function of $\l$ is shown in Figure~\ref{ggssDA120920Fdsv4}(c).

Next we deform the system by $F^s(\l)$ for  $\l \in [-0.25,0.25]$:  $\bar{\g}_1^s = 0.8970$, $\bar{\g}_{1+2}^s = \bar{\g}_{1-2}^s = 0.8020$ minimises the least-square relative error over all $\l$ in the chosen range. Figure~\ref{ggssDA120920Fdsv4}(b) compares the true energy with the energy given by homogenisation for each $\l$, and the fractional error of the homogenized energy to the true energy as a function of $\l$ is shown in Figure~\ref{ggssDA120920Fdsv4}(d).

\end{simulation}

\begin{figure}
\begin{center}
\includegraphics[width=0.75\textwidth]{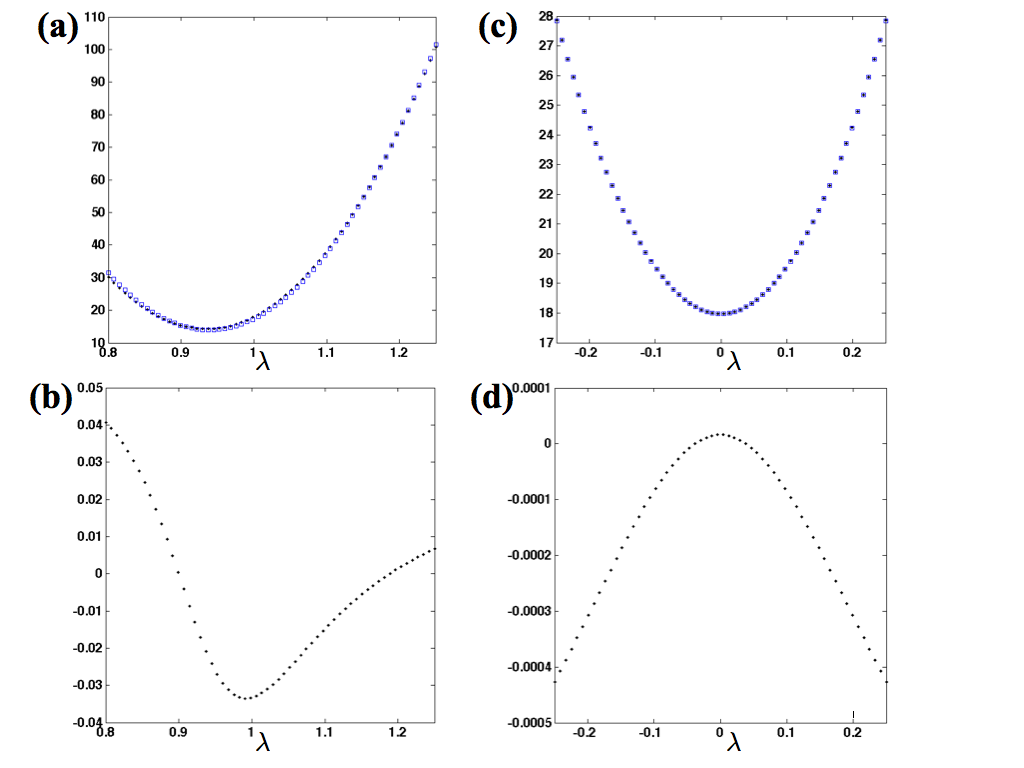}
\caption{Results of Simulation~\ref{sim:Baskin}.  The true energy (black dots) and the homogenized energy (blue squares) are plotted as functions of $\lambda$ for deformations of the form (a) $F^d(\l)$ and (c) $F^s(\l)$.  The fractional error of the homogenized to true energy are plotted in (b) for the results in (a) and in (d) for the results in (c).}
\label{ggssDA120920Fdsv4}
\end{center}
\end{figure}

In Simulation~\ref{sim:Baskin}, the tolerance associated with the choice of $\bar{G}_1,\bar{G}_2$ is $4 \%$ error; in the simulations that follow the maximum error is about $10\%$.  Table~\ref{table} gives computations of $\bar{\g}_1^d$ and $\bar{\g}_{1+2}^d, \bar{\g}_{1-2}^d$ and the mean squared error (\ref{mse}) for simulations similar to Simulation~\ref{sim:Baskin}, but with grids $(\Z^2_N,\{e_1,e_2,e_{1\pm2} \})$ of varying size.

\begin{table}%[h]
\centering
\begin{tabular}{|c|c|c|c|}
$N$ & $\bar{\g}_1^d$ & $\bar{\g}_{1+2}^d = \bar{\g}_{1-2}^d$ & \text{mean squared error (\ref{mse}) for $(\Z^2_N,\{e_1,e_2,e_{1\pm2} \})$} \\
\hline
$2^3$ & 1.2050 & 0.8270 & 0.1405 \\
$2^4$ & 1.2010 & 0.8300 & 0.1676 \\
$2^5$ & 1.1970 & 0.8300 & 0.1618 \\
$2^6$ & 1.1980 & 0.8300 & 0.1600 \\  
\hline
\end{tabular}
\caption{Computations of $\bar{\g}_1^d$ and $\bar{\g}_{1+2}^d, \bar{\g}_{1-2}^d$ for simulations identical to Simulation~\ref{sim:Baskin} except for grid size. The mean squared error (\ref{mse}) is calculated over 60 equidistant choices of $\lambda$ between 0.8 and 1.25.}
\label{table}
\end{table}
 
\begin{remark} Simulation~\ref{sim:Baskin} is related to a situation that occurs in the growth of plant cells; see~\cite{Baskin:2005}. For a cluster of neighbouring cells with the same cellulose microfibril orientation, the growth of the plant cell is related to the microfibril orientation. However when there is a large variation in microfibril orientation in a cluster of cells, then microfibril orientation of a cell is less correlated with the growth of that cell. In this case, as in Simulation~\ref{sim:Baskin}, the homogenised growth of the cell closer is related to the homogenised microfibril orientation.
\end{remark}

\begin{simulation}[$(\Z^2_{16},\{e_1,e_2,e_{1\pm2} \})$ with growth in all springs uniformly distributed about 1] \label{sim:S0v4}  \par \indent

\textbf{Initial state:} $L_1 = L_2 = 1$ and $L_{1+2} = L_{1-2} = \sqrt{2}$. Note that the rest-state has zero energy, that is, is residual-stress free.

\textbf{Recombination}: $g_1,g_2,g_{1+2},$ and $g_{1-2}$ were all chosen from uniform distributions on $[0.8,1.2]$.  In contrast to Simulation \ref{sim:Baskin}, generically
\begin{equation*}
 \left( \sqrt{2} g_{1+2} \right)^2 + \left( \sqrt{2} g_{1-2} \right)^2 \neq 2 \left( g_1^2 + g_2^2 \right),
\end{equation*}
so unit cells of the recombined lattice generically do not have a zero-energy state. Thus growth leads to residual stress. 

\textbf{Result}: We choose $W_1$ and $W_2$ as in~\eqref{eq:energy-deformation1}, with $q=2$ in \eqref{Wqpick}. We  choose $\overline{G}_1$  to be of the form $\bar{\g}_1 I$ and $\overline{G}_2$ to be of the form $\bar{\g}_2 I $.  

Let us first impose a deformation $F^d(\l)$ for  $\l \in [1/1.5,1.5]$: $\bar{\g}_1^d = 1.0240$, $\bar{\g}_2^d= 0.8600$ minimises the least-square relative error over all $\l$ in the chosen range. Figure~\ref{ggssHVDS0120918FdsUv2}(a) compares the true energy with the energy given by homogenisation for each $\l$, and the fractional error of the homogenized energy to the true energy as a function of $\l$ is shown in Figure~\ref{ggssHVDS0120918FdsUv2}(c).

Next we deform the system by $F^s(\l)$ for  $\l \in [-0.5,0.5]$:   $\bar{\g}_1^s = 1.050$ and $\bar{\g}_2^s=0.993$ minimises the least-square relative error over all $\l$ in the chosen range. Figure~\ref{ggssHVDS0120918FdsUv2}(b) compares the true energy with the energy given by homogenisation for each $\l$ , and the fractional error of the homogenized energy to the true energy as a function of $\l$ is shown in Figure~\ref{ggssHVDS0120918FdsUv2}(d).

Repeating this simulation, but choosing $g_1, g_2$  from uniform distributions on $[1-\delta g_{hv},1+\delta g_{hv}]$ and $g_{1+2}, g_{1-2}$ from uniform distributions on $[1-\delta g_d,1+\delta g_d]$, we plot the resulting homogenization factors $\bar{\g}_1^d$ and  $\bar{\g}_2^d$  in Figure~\ref{ggssHVDS0120918FdsUv2n2}. 

Results using $q=3$ ($\bar{\g}_1^d = 1.0100$, $\bar{\g}_2^d= 0.9590$) and $q=4$ ($\bar{\g}_1^d = 1.0060$, $\bar{\g}_2^d= 0.9640$) and deformations $F^d$ are shown in Fig.~\ref{5p5v34}.  

\end{simulation}

\begin{figure}
\begin{center}
\includegraphics[width=0.75\textwidth]{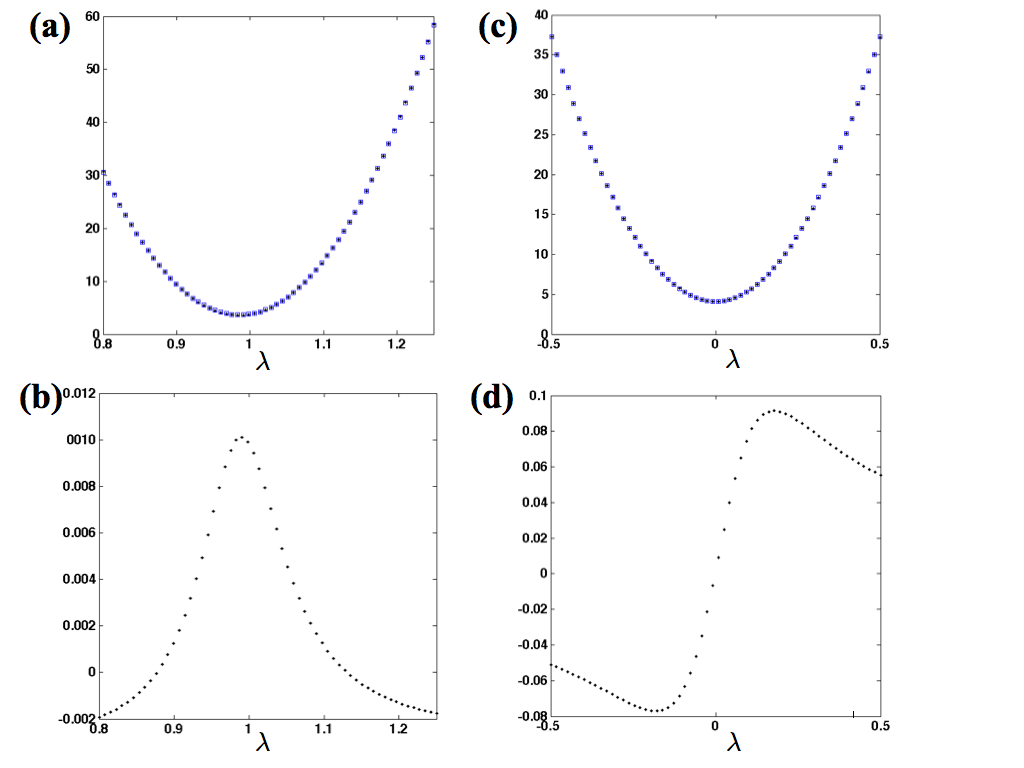}
\caption{Results of Simulation~\ref{sim:S0v4} for $q=2$.   The true energy (black dots) and the homogenized energy (blue squares) are plotted as functions of $\lambda$ for deformations of the form (a) $F^d(\l)$ and (c) $F^s(l)$.  The fractional error of the homogenized to true energy are plotted in (b) for the results in (a) and in (d) for the results in (c).}
\label{ggssHVDS0120918FdsUv2}
\end{center}
\end{figure}

\begin{figure}
\begin{center}
\includegraphics[width=0.5\textwidth]{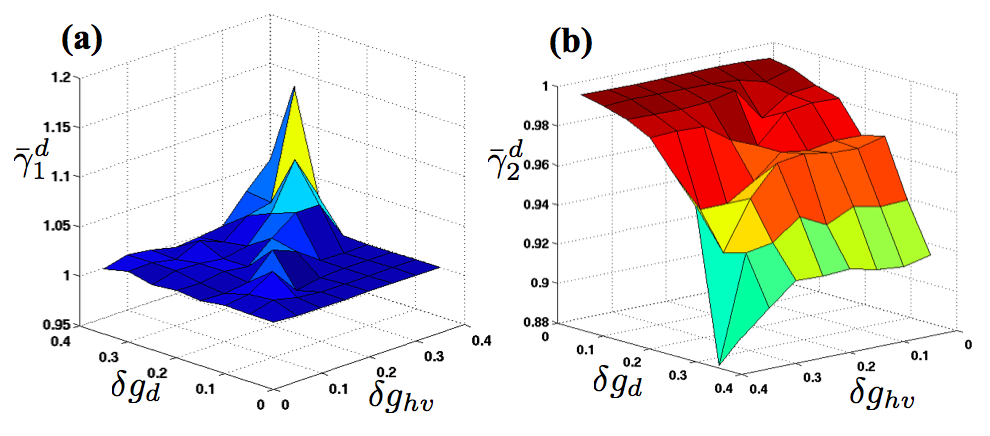}
\caption{Results of Simulation~\ref{sim:S0v4}: (a) $\bar{\g}_1^d$ and (b) $\bar{\g}_2^d$  as functions of $\delta g_{hv}$ and $\delta g_d$ (see text).}
\label{ggssHVDS0120918FdsUv2n2}
\end{center}
\end{figure}

\begin{figure}
\begin{center}
\includegraphics[width=0.75\textwidth]{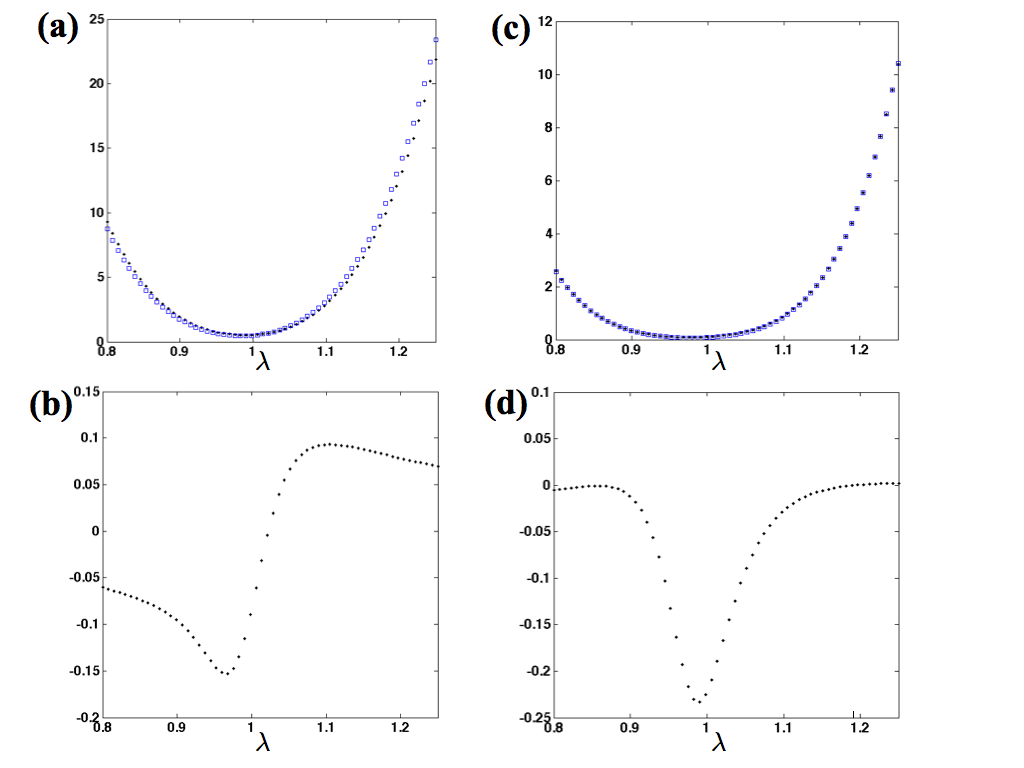}
\caption{Results of Simulation~\ref{sim:S0v4} for (a,b) $q=3$ and (c,d) $q=4$.  The true energy (black dots) and the homogenized energy (blue squares) are plotted as functions of $\lambda$ for deformations of the form $F^d(\l)$.  The fractional error of the homogenized to true energy are plotted in (b) for the results in (a) and in (d) for the results in (c).}
\label{5p5v34}
\end{center}
\end{figure}

\begin{simulation}[$(\Z^2_{16},\{e_1,e_2,e_{1\pm2} \})$ with residual stress; growth in all springs uniformly distributed about 1] \label{sim:S0v4vRS}  \par \indent

\textbf{Initial state:} $L_1 = L_2 = 1$ and $L_{1+2} = L_{1-2} = 1.25$. Note that the rest-state has non-zero energy, that is, it has residual stress.

\textbf{Recombination}: $g_1,g_2,g_{1+2},$ and $g_{1-2}$ were all chosen from uniform distributions on $[0.8,1.2]$. Since, generically,
\begin{equation*}
 \left( 1.25 g_{1+2} \right)^2 + \left( 1.25 g_{1-2} \right)^2 \neq 2 \left( g_1^2 + g_2^2 \right),
\end{equation*}
generically, unit cells of the recombined lattice do not have a zero-energy state. Thus the grown system is also residually stressed. 

\textbf{Result}: We choose $W_1$ and $W_2$ as in~\eqref{eq:energy-deformation1}, with $q=2$ in \eqref{Wqpick}. We  choose $\overline{G}_1$  to be of the form $\bar{\g}_1I$ and $\overline{G}_2$ to be of the form $\bar{\g}_2 I $.  

Let us first impose a deformation $F^d(\l)$ for  $\l \in [1/1.5,1.5]$: $\bar{\g}_1^d = 1.132$, $\bar{\g}_2^d= 0.957$ minimises the least-square relative error over all $\l$ in the chosen range. Figure~\ref{ggssHVDS0120918FdsUvRS}(a) compares the true energy with the energy given by homogenisation for each $\l$, and the fractional error of the homogenized energy to the true energy as a function of $\l$ is shown in Fig.~\ref{ggssHVDS0120918FdsUvRS}(c).

Next we deform the system by $F^s(\l)$ for  $\l \in [-0.5,0.5]$: $\bar{\g}_1^s = 1.007$ and $\bar{\g}_2^s=1.085$ minimises the least-square relative error over all $\l$ in the chosen range. Fig.~\ref{ggssHVDS0120918FdsUvRS}(b) compares the true energy with the energy given by homogenisation for each $\l$ , and the fractional error of the homogenized energy to the true energy as a function of $\l$ is shown in Fig.~\ref{ggssHVDS0120918FdsUvRS}(d).

\end{simulation}

\begin{figure}
\begin{center}
\includegraphics[width=0.75\textwidth]{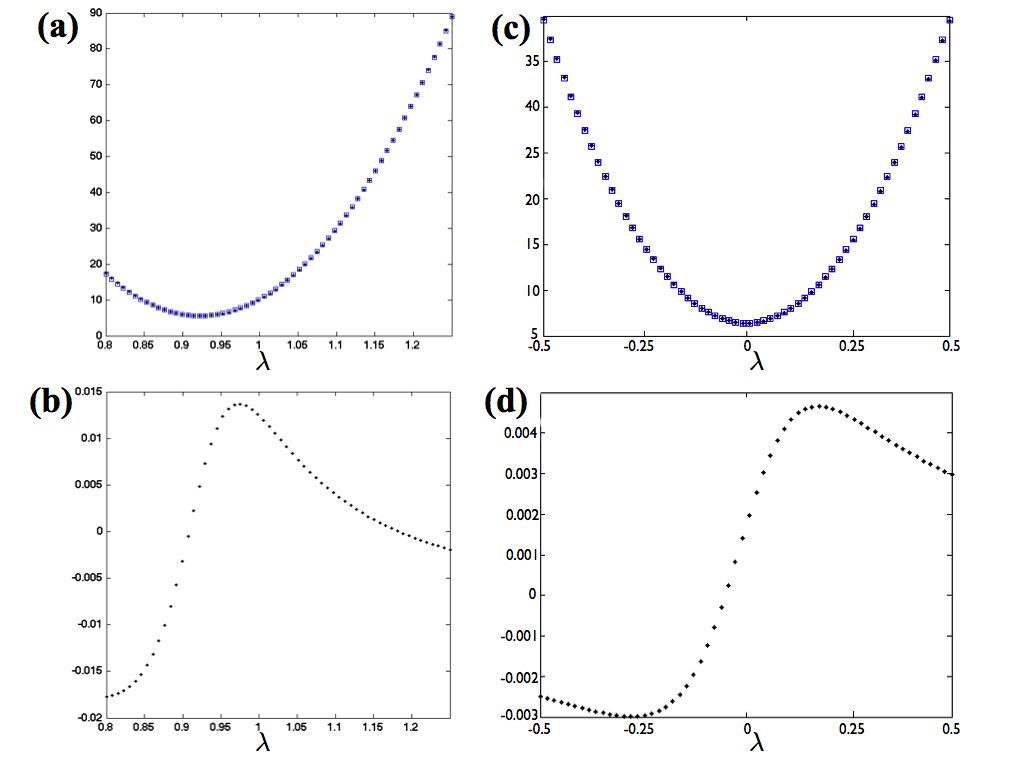}
\caption{Results of Simulation~\ref{sim:S0v4vRS}.  The true energy(black dots) and the homogenized energy (blue squares) are plotted as functions of $\lambda$ for deformations of the form (a) $F^d(\l)$ and (c) $F^s(l)$.  The fractional error of the homogenized to true energy are plotted in (b) for the results in (a) and in (d) for the results in (c).}
\label{ggssHVDS0120918FdsUvRS}
\end{center}
\end{figure}

%-------------
\subsection{Inhomogenenous $(\Z^2_N,\{e_1,e_2,e_{1\pm2} \})$ recombining inhomogenenously}
\label{sec:inhomogeneous_lattices}

In the simulations of Section~\ref{sec:homogeneous_lattices_recombining_inhomogeneously} we were guided by Lemma~\ref{lem:!unique_decomposition} in our choice of the the energies $W_1$ and $W_2$ in the decomposition. For all three simulations we picked these as in~\eqref{eq:energy-deformation1}, but~\eqref{eq:energy-deformation2} and~\eqref{eq:energy-deformation3} would have served just as well.

When the initial lattice is inhomogeneous we determine $W_1$ and $W_2$ from the homogeneous representative of the initial lattice. Such a homogeneous lattice has a continuum energy of the form 
\begin{equation*}
 \sum_{ \alpha \in \{1,2,1+2,1-2\} } W \left( \frac{\| F e_\alpha \|}{\overline{L}_\alpha} \right).
\end{equation*}
and thus it remains to find the optimal homogenised rest-lengths $\overline{L}_\alpha$, $\alpha \in \{1,2,1+2,1-2\}$. We do this by picking these to minimise the relative mean square error
\begin{equation*}
\sum_{F \in \F} \left( \frac{W_i(F) - \left( \sum_{ \alpha \in \{1,2,1+2,1-2\} } W \left( \frac{\| F e_\alpha \|}{\overline{L}_\alpha} \right) \right)}{W_i(F)} \right)^2
\end{equation*}
(where the outer sum is approximated by a uniform sampling over $\F$). We accept these to be the desired rest-lengths provided that
\begin{enumerate}
\item The error is less than a tolerance, and
\item the rest-lengths are convergent, and the error is non-increasing, when the grid size increases. 
\end{enumerate}

We now allow the lattice to recombine (inhomogeneously) and proceed as in Section~\ref{sec:homogeneous_lattices_recombining_inhomogeneously}.

\begin{simulation}[$(\Z^2_{16},\{e_1,e_2,e_{1\pm2} \})$ with nonuniform rest lengths; growth in all springs uniformly distributed about 1] 
\par \indent

\textbf{Initial state:} $L_1$ and  $L_2 $ were chosen from uniform distributions on $[0.8,1.2]$.  $L_{1+2}$ and $L_{1-2}$ were chosen from uniform distributions on $[0.8\sqrt{2},1.2\sqrt{2}]$. Note that the rest-state has non-zero energy, that is, it has residual stress.

\textbf{Recombination}: $g_1,g_2,g_{1+2},$ and $g_{1-2}$ were all chosen from uniform distributions on $[0.8,1.2]$. Since, generically,
\begin{equation*}
 \left( 1.25 g_{1+2} \right)^2 + \left( 1.25 g_{1-2} \right)^2 \neq 2 \left( g_1^2 + g_2^2 \right),
\end{equation*}
generically, unit cells of the recombined lattice do not have a zero-energy state. Thus the grown system is also residually stressed. 

\textbf{Result}: We choose $W_1$ and $W_2$ as in~\eqref{eq:energy-deformation1}, with $q=2$ in \eqref{Wqpick}.  

Let us first impose deformations $F^d(\l)$ for  $\l \in [1/1.5,1.5]$:  Homogenised rest lengths $\bar{L}_1 = \bar{L}_2=1.0820$ and $\bar{L}_{1+2} = \bar{L}_{1-2} = 0.8790$ minimise the least-square relative error over all $\l$ in the chosen range (with 4$\%$ maximum fractional error). Allowing the springs to grow, homogenised growths of the form $\overline{G}_1=\bar{\g}_1I$ and $\overline{G}_2=\bar{\g}_2 I $ with $\bar{\g}_1^d = 1.0560$ and $\bar{\g}_2^d = 1.0330$ minimise the least-square relative error over all $\l$ in the chosen domain (with 5$\%$ maximum fractional error). 

Next we  deform the system by $F^s(\l)$ for  $\l \in [-0.5,0.5]$:  Homogenised rest lengths $\bar{L}_1 = \bar{L}_2=1.0260$ and $\bar{L}_{1+2} = \bar{L}_{1-2} = 0.8530$ minimise the least-square relative error over all $\l$ in the chosen range (with 2.5$\%$ maximum fractional error). Allowing the springs to grow, homogenised growths of the form $\overline{G}_1=\bar{\g}_1I$ and $\overline{G}_2=\bar{\g}_2 I $ with $\bar{\g}_1^d = 1.1930$ and $\bar{\g}_2^d = 0.8110$ minimise the least-square relative error over all $\l$ in the chosen domain (with 6$\%$ maximum fractional error). 
\end{simulation}

%---------
\section{Outlook: Recombining networks and continua}
\label{sec:continua}

While the precise analysis above is limited to lattices, the central insights regarding multiplicative decompositions and energy-deformation decompositions hold for networks in general. (By a network we mean a system of springs connected through nodes, a lattice being a special case.)

This follows from the observation that in a $D$-dimensional network that is able to resist shears a \emph{generic} node should be connected to at least $D+1$ springs, see Figure~\ref{fig:networks}. It follows, from the reasoning in Section~\ref{sec:homogeneous_lattices_recombining_homogeneously}, that the continuum energy of such a network cannot obey a multiplicative decomposition.

\begin{figure}
\begin{center}
\begin{subfigure}{0.3\textwidth}
\includegraphics[width=\textwidth]{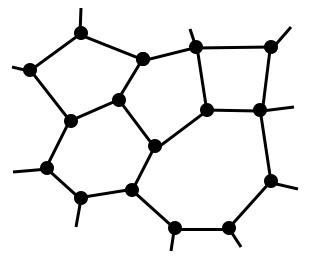}
\caption{}
\label{fig:network2}
\end{subfigure}
\begin{subfigure}{0.3\textwidth}
\includegraphics[width=\textwidth]{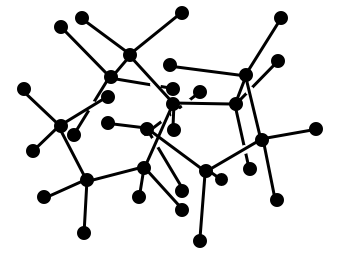}
\begin{center}  \end{center}
\caption{}
\label{fig:network3}
\end{subfigure}
\caption{Shear-resisting networks in two and three dimensions where each node is minimally connected.}
\label{fig:networks}
\end{center}
\end{figure}

On the other hand when each node of a network is connected to no more than $2D$ other nodes then, reasoning as in the proof of Theorem~\ref{thm:Z^D}, we are led to hope that two energy-densities and two growth descriptors would describe the continuum energy-density of the network. Thus we hypothesise:

\begin{hypothesis}[$D$-dimensional shear-resisting networks with connectivity no more than $2D$]
\label{hyp:recombining_networks}
The initial continuum energy-density $W_i$ of a $D$-dimensional shear-resisting network with connectivity no more than $2D$ is related to its recombined continuum energy-density $W_g$ through
\begin{align}
W_i(\cdot) &= W_1(\cdot) + W_2(\cdot), \tag{\ref{eq:energy-deformation-initial}} \\
W_g(\cdot) &= W_1(\cdot G_1^{-1}) + W_2(\cdot G_2^{-1}). \tag{\ref{eq:energy-deformation-grown}}
\end{align}
for some $W_1, W_2 \colon \Def \to \R$, each of which is independent of growth and vanishes on a shear, and some $G_1,G_2 \in \Def$ which are independent of the deformation.
\end{hypothesis}

From~\eqref{eq:energy-deformation},
\begin{equation*}
W_g(F)
 = W_i(F G_2^{-1}) + \left( W_1(F G_1^{-1}) - W_1(F G_2^{-1}) \right).
\end{equation*}
It is convenient to set
\begin{align}
H &:= G_2^{-1}, \notag \\
H' &:= G_2 G_1^{-1}, \notag \\
W'(F,\cdot) &:= W_1(F\cdot) - W_1(F). \label{eq:energy_correction}
\end{align}
This gives,
\begin{equation} \label{eq:continuum_energy}
W_g(F)
 = W_i(FH) + W'(FH,H').
\end{equation} 

(Alternatively, from~\eqref{eq:energy-deformation},
\begin{equation*}
W_g(F)
 = W_i(F G_1^{-1}) + \left( W_2(F G_2^{-1}) - W_2(F G_1^{-1}) \right).
\end{equation*}
and with the alternate definitions
\begin{align*}
H &:= G_1^{-1}, \\
H' &:= G_1 G_2^{-1}, \\
W'(F,\cdot) &:= W_2(F\cdot) - W_2(F).
\end{align*}
we again obtain~\eqref{eq:continuum_energy}.)

Observe that $H$ and $H'$ are measures of growth with the later measuring the extent to which growth is anisotropic. $W'$ measures the extent to which the energy of the grown body differs from a motion of the initial energy. When the growth is isotropic then $H'=I$ and $W'(\cdot,I) = 0$ thus recovering the multiplicative decomposition as a special case.

We investigate this formulation further in~\cite{Chenchiah-Shipman2}. There we also investigate the question of which of the energy densities that have been proposed for soft tissues (see, for example, \cite{Bogen:1979,Chaplain:1993,Delfino:1997,Holzapfel:2000,Sacks:2000,Horgan:2002p353,Horgan:2003,Shergold:2006}) are compatible with Hypothesis~\ref{hyp:recombining_networks}. One natural way to pursue this question is via the following definition:

\begin{definition}[Shear decomposability]
$W \colon \Def \to \R$ is shear decomposable of order $K$ if there exist $W_i \colon \Def \to \R$, $i=1,\dots,K$, such that
\begin{enumerate}
\item Each $W_i$, $i=1,\dots, K$, vanishes on a shear, and
\item $W = \sum_{i=1}^K W_i$.
\end{enumerate}
When $K=2$ we shall omit ``of order two'' and say ``shear decomposable''. 
\end{definition}

The significance of this definition is, of course, the following corollary of Hypothesis~\ref{hyp:recombining_networks}:

\begin{corollary}
Let $W \colon \Def \to \R$ be the energy-density of a recombining $D$-dimensional shear-resisting network with connectivity no more than $2D$. Then $W$ is shear decomposable.
\end{corollary}

So the question can be rephrased:

\begin{question}
Which of the energy densities that have been proposed for soft tissues are shear decomposable?
\end{question}

\bigskip

\noindent \textbf{Acknowledgement:} PS was supported in part by National Science Foundation grant DMS-1022635. This work grew out of discussions at the Max-Planck-Institute for Mathematics in the Sciences, Leipzig. We thank Stefan Mueller for encouragement to work on the mechanics of biological growth, and Alain Goriely for encouraging us to pursue this line of research.

\bibliographystyle{elsarticle-harv}
\bibliography{Papers,Manuscripts}

\end{document}